\documentclass[preprint,12pt]{elsarticle}
\usepackage[margin=1in]{geometry} 
\usepackage{amsmath}
\usepackage[utf8]{inputenc}
\usepackage{hyperref}
\usepackage[english]{babel}
\usepackage{placeins}
\usepackage[usenames, dvipsnames]{color}
\usepackage{algorithm}
\usepackage[noend]{algpseudocode}
\makeatletter
\def\BState{\State\hskip-\ALG@thistlm}
\makeatother
\usepackage{float}
\usepackage{rotating}
\usepackage{booktabs}
\usepackage{comment}
\usepackage{xcolor}
\usepackage[colorinlistoftodos]{todonotes}
\usepackage{tikz}
\usepackage{mathtools}
\usetikzlibrary{shapes}
\usetikzlibrary{plotmarks}

\usepackage{xcolor}

\usepackage[most]{tcolorbox}
\newtcolorbox{highlighted}{colback=yellow,coltext=red,breakable}
\usepackage{fontawesome} % For images icons

\usepackage{nomencl}
\makenomenclature
\usepackage{footnote}
\usepackage[flushleft]{threeparttable}
\usepackage{graphicx}
\usepackage{amssymb}
\usepackage{cleveref}
\usepackage{hyperref}
\usepackage[utf8]{inputenc}
\usetikzlibrary{decorations.pathreplacing}

\usepackage{lineno}
\usepackage[utf8]{inputenc}
\PassOptionsToPackage{draft}{graphicx}
\usepackage[draft]{graphicx}
\usepackage{subcaption}
\usepackage{nomencl}
\usepackage{comment}
\usepackage{algpseudocode}
\usepackage{natbib}
\setcitestyle{authoryear,open={},close={}}
\usepackage{amsthm}
\newtheorem{prop}{Proposition}
\newtheorem{rem}{Remark}
\newtheorem{defn}{Definition}
\algnewcommand\And{\textbf{and }}
\algnewcommand\Or{\textbf{or }}
\usepackage{cleveref}
\crefformat{section}{\S#2#1#3} % see manual of cleveref, section 8.2.1
\crefformat{subsection}{\S#2#1#3}
\crefformat{subsubsection}{\S#2#1#3}
\usepackage{multirow}
\journal{}

\begin{document}
	
	\begin{frontmatter}
		
		
		\title{An algorithm for integrating peer-to-peer ridesharing and schedule-based transit system for first mile/last mile access}
		

		
		\author[pk]{Pramesh Kumar}
		\author[pk]{Alireza Khani*}
		\address[pk]{Department of Civil, Environmental and Geo-Engineering, University of Minnesota, Twin Cities}

		\begin{abstract}
			
        Due to limited transit network coverage and infrequent service, suburban commuters often face the transit first mile/last mile (FMLM) problem. To deal with this, they either drive to a park-and-ride location to take transit, use carpooling, or drive directly to their destination to avoid inconvenience. Ridesharing, an emerging mode of transportation, can solve the transit first mile/last mile problem. In this setup, a driver can drive a ride-seeker to a transit station, from where the rider can take transit to her respective destination. The problem requires solving a ridesharing matching problem with the routing of riders in a multimodal transportation network. We develop a transit-based ridesharing matching algorithm to solve this problem. The method leverages the schedule-based transit shortest path to generate feasible matches and then solves a matching optimization program to find an optimal match between riders and drivers. The proposed method not only assigns an optimal driver to the rider but also assigns an optimal transit stop and a transit vehicle trip departing from that stop for the rest of the rider’s itinerary. We also introduce the application of space-time prism (STP) (the geographical area which can be reached by a traveler given the time constraints) in the context of ridesharing to reduce the computational time by reducing the network search. An algorithm to solve this problem dynamically using a rolling horizon approach is also presented. We use simulated data obtained from the activity-based travel demand model of Twin Cities, MN to show that the transit-based ridesharing can solve the FMLM problem and save a significant number of vehicle-hours spent in the system. 
			
		\end{abstract}
		
				
		\cortext[mycorrespondingauthor]{Corresponding author}
		\fntext[myfootnotetel]{Email: akhani@umn.edu}
		\fntext[myfootnotetel]{Tel: (612) 624-4411}
		\fntext[myfootnotetel]{Web: \href{http://umntransit.weebly.com/}{http://umntransit.weebly.com/}}

	\end{frontmatter}
	%\linenumbers
	\newpage
	\section{Introduction}\label{sec:intro}
    Due to a rise in the number of motor vehicles resulting from increasing travel demand, urban highways are facing an inescapable condition of traffic congestion. The congestion on roads can be primarily ascribed to an increase in the use of a personal vehicle for traveling. Besides, we have observed that the vehicle capacity is often under-utilized, e.g., National Household Travel Survey 2009 shows that on an average, only 1.7 out of 4 available seats in cars were utilized (\citet{Santos2009}, \citet{Masoud2017b}). As the aim is to transport people, not cars, an alternative way to provide mobility is to make efficient use of existing transportation infrastructure. Public transportation, which can carry multiple passengers, is widely considered as a practical solution to the congestion problem by reducing vehicle-miles traveled (VMT) on roads (\citet{Aftabuzzaman2015}). 
	
		\begin{figure}[h!]
		\centering
		\begin{tikzpicture}
		\draw[-, thick] (0, 0)--(15, 0) node[right]{};
		\draw[-, thick] (3, 0.2)--(3, -0.2) node[right]{};
		\draw[-, thick] (6, 0.2)--(6, -0.2) node[right]{};
		\draw[-, thick] (9, 0.2)--(9, -0.2) node[right]{};
		\draw[-, thick] (12, 0.2)--(12, -0.2) node[right]{};
		\node[text width=4cm] at (2.4, 0.5) {Access time};	
		\node[text width=4cm] at (2.4, -0.5) {(First mile)};	
		\node[text width=4cm] at (5.3, 0.5) {Waiting time};	
		\node[text width=4cm] at (8.2, 0.5) {In-vehicle time};
		\node[text width=4cm] at (11.3, 0.5) {Transfer time};	
		\node[text width=4cm] at (14.4, 0.5) {Egress time};	
		\node[text width=4cm] at (14.4, -0.5) {(Last mile)};	
		\end{tikzpicture}
		\caption{Components of a tansit trip}
		\label{fig:transittrip}
	\end{figure}
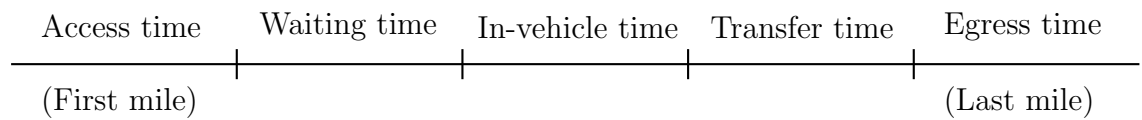
    To use public transit for travel, one has to walk from origin to a boarding stop to access the service and then walk again from alighting stop to the respective destination.  Formally, these walking components of a transit trip are known as "first mile" and "last mile," which describe the beginning and end of an individual transit trip (\figurename{\ref{fig:transittrip}}) respectively. Due to limited transit coverage, it is sometimes difficult or impossible to walk to/from transit stop on either end of a transit trip. This inaccessibility problem is also known as the first mile/last mile (FMLM) problem for transit. The problem is common among people commuting from low-density areas where transit service is not available or less frequent. The limited coverage in these areas is because of the economic in-viability of providing transit service. To encourage people living in suburban areas to use transit to commute, transportation agencies need to provide an effective transportation service that can help them to cover their first/last mile. Various possible solutions to this problem are driving to a park-and-ride facility, use para-transit services, taxi/carpooling, or bike to a transit station. Although these solutions to the FMLM problem are feasible, they are either not attractive among commuters or too expensive for daily travel, which drives travelers to use personal vehicle for their daily commute. A cheaper solution is to use the existing trips as a feeder service to transit stations. Ridesharing, which has attracted attention during recent years, can be a potential solution to the transit FMLM problem. In this program, both riders and drivers can submit a request in real-time, and an automated system will match them for sharing a ride while satisfying the spatial and temporal constraints for both.\\
	
    The ridesharing programs are becoming popular in cities (\citet{Agatz2011, Masoud2017}). However, a few studies such as \citet{Masoud2017} have raised a concern that this increase in the participation in ridesharing programs is shifted from the transit ridership. Besides, \citet{Wang2018}, using a ridesharing equilibrium model, showed that when driving is faster and more expensive than transit, then no cost-sharing strategy can sustain a ridesharing program without shifting the demand from transit mode. In that case, ridesharing, which was aimed to reduce congestion on roads, may not be as effective as anticipated, and it will become a competitor of transit mode. However, the equilibrium conditions proved in that study is a function of the matching probability and does not consider the integration of ridesharing and transit mode. A ridesharing program integrating transit with ridesharing can help in shifting the share from single-occupancy vehicles and hence increase transit ridership. In his program, ridesharing can provide a fast and reliable mobility service in low-density areas, and transit can provide mobility in high-density areas.  \\
    
    In this article, we develop a framework to integrate ridesharing with transit to solve the FMLM problem.  Riders who are looking for an affordable ride to a transit station and drivers who want to join this program to alleviate the negative impacts of car travel and to receive a small compensation can submit their request to an automated system. The framework consists of the development of an algorithm for transit passenger routing in a schedule-based transit network and a matching optimization program to match riders and drivers up to the first mile of the transit trip. Overall, this article makes the following contributions:

	\begin{enumerate}	
        \item Propose operation of a large scale transit-based ridesharing program,
	\item Development of a schedule-based transit network algorithm for finding feasible matches,
	\item Incorporating the state-of-the-art network search reduction techniques to improve the computational time without compromising the solution quality,
	\item Development of a \textit{rolling horizon} algorithm to solve the dynamic transit ridesharing problem,
	\item Use simulation to assess the efficiency of the proposed method on a large scale transit network of Twin Cities, MN    
	\end{enumerate}

    The rest of the article is structured as follows. \cref{sec:literature} reviews previous work on transit FMLM problem and ridesharing matching algorithms, which is followed by the motivation for this research. \cref{sec:prelim} introduces the terms and definitions used in this article and describes the operation of the transit-based ridesharing program. Then, \cref{sec:tbrp} formally defines the transit-based ridesharing matching problem. After that, an algorithm to find potential matches and a matching optimization model is developed in \cref{sec:sbtsp}, which is followed by the development of a rolling horizon algorithm. Then, \cref{sec:exp} presents the results of simulation experiments and \cref{sec:disc} describes ways to relax various assumptions in this study. Finally, conclusions and recommendations for future research are given in \cref{sec:conc}.

	\section{Related work and motivation for this study}\label{sec:literature}
        Previous studies have found that access to transit is a significant factor affecting its modal share (\citet{Moniruzzaman2012, Brons2009, Kalaanidhi2013}). Therefore, in areas with limited network coverage, it is crucial to provide a feeder service to attract riders to use transit service. Various modes of transport have been studied and implemented to improve access to transit mode. This includes integrating bikesharing and transit (\citet{Martin2014, Rietveld2000, martens2004bicycle}), designing a demand responsive transit feeder service to solve the FMLM problem (\citet{Wang2017a, maheo2017benders, Cayford2004, Koffman2004, Lee2017, Quadrifoglio2008, Shen2012, Li2009a}), use of park-and-ride facilities (\citet{Nassir2012, Khani2012, Khani2020}), and integrating ridesharing and transit (\citet{Masoud2017, Stiglic2018, Bian2019, Ma2019, Chen2020}). We summarize the literature related to transit-based ridesharing into three main categories.
	    
	    \begin{enumerate}
            \item \textit{Designing a feeder service to the transit mode}: The transit-based ridesharing problem resembles similarities with the literature focused on designing a feeder service to the transit mode such as designing a demand responsive connector (DRC) (\citet{Cayford2004, Koffman2004, Lee2017, Quadrifoglio2008, Shen2012, Li2009a}) or dial-a-ride service with transfers (DART) (\citet{Masson2014, Deleplanque2013, Jafari2016, Hall2009}). \citet{Cayford2004} designed a demand-responsive dial-a-ride service in California.  The service was designed to be fixed-route during high demand and flexible-route during low demand period. A transit cooperative research program (TCRP) report also presents experiences from flexible transit services in various cities in the United States (\citet{Koffman2004}). While designing vehicle-scheduling for a demand responsive system, many studies consider zoning in which a vehicle would operate within the zonal boundary to pick up the passengers and drop them off at a predetermined transit station (\cite{Quadrifoglio2008, Shen2012, Li2009a}). \citet{Maheo2015} proposed the design of a hub and shuttle public transit system in Canberra. \citet{Wang2017a}  proposed a mixed-integer linear program and heuristic techniques to design routing and scheduling for a last-mile transportation service. However, the transit-based ridesharing problem we consider in this article, is different from the ones discussed above in the following aspects. First, in ridesharing, the drivers are not hired to work for a private entity, so the departure and arrival time constraints of both drivers and riders need to be incorporated in the model. Second, we cannot assume an unlimited supply of vehicles to serve rider trips (\cite{Lee2017}). Third, ridesharing problem is highly dynamic in nature, where drivers and riders enter the system on short notice, it is difficult to solve a possible discrete optimization model proposed in the studies cited above in real-time.

	    	\item \textit{Ridesharing matching algorithms}: The current study proposes a matching algorithm for peer-to-peer (P2P) ridesharing as an alternative mode of transportation to access transit. The ridesharing matching problem has recently got the attention of many researchers. \citet{Agatz2012} and \citet{Furuhata2013a} reviewed various forms of ridesharing and outlined the difficulties that arise when developing a mechanism for such a system. \citet{Agatz2011} developed a simple optimization model for matching single driver and rider while satisfying their time budget constraints and keeping the driver or rider role flexible. \citet{Stiglic2015} pointed out the benefits of meeting points in ridesharing. \citet{Santi2014} proposed the idea of a shareability network to reduce the size of the matching problem. As the method was only able to match two riders (optimally) and three riders (heuristically), \citet{Alonso-mora2017} improved this idea by creating a request-to-vehicle (RV) and a request-to-trip-to-vehicle (RTV) graph to solve a high capacity ridesharing problem. \citet{Masoud2017b, Masoud2017a} presented a solution to the general peer-to-peer multiple matching problem with transfers in a time-dependent network. They proposed a decomposition algorithm and a dynamic programming approach to solve the problem and showed an increase in the total number of matching by using their algorithm. \\
	    	
            \item \textit{Integrating ridesharing with transit}: To integrate ridesharing with transit, \citet{Masoud2017} used the same dynamic programming framework as in \citet{Masoud2017a} to solve the current problem by introducing the transit stops and transfers as go-to-points and transit route as an inflexible driver in their algorithm. They considered one transit route and limited the number of transfer points to 40 to reduce the computational time. The major obstacle to applying their algorithm to a large-scale transit system would be high computational time.  \citet{Ma2019} proposed various queueing-theoretic vehicle dispatch and idle vehicle relocation algorithms for ridesharing with a possible drop-off of a passenger at a transit stop. \citet{Bian2019, Bian2019a} developed a matching optimization program and a Vickrey-Clarke-Groove (VCG) mechanism to determine the optimal vehicle-passenger matching, vehicle routing plan, and a customized pricing scheme respectively. The proposed mechanism is proved to be individual rational, incentive compatible, and price non-negative. \citet{Chen2020} developed a mixed-integer linear programming (MILP) to solve the FMLM with autonomous vehicles. A pioneering effort in developing a matching algorithm for this multimodal system is described by \citet{Stiglic2018}. They presented a matching framework to match riders and drivers for transit, park-and-ride, or ridesharing mode. However, they considered a cyclic (frequency-based) time table for transit service, assumed the closest stop to the destination as the alighting stop, and selected a transit trip with the least driving distance to make the complex problem easier to solve. 
	    \end{enumerate}

	    	\subsection{Motivation} The following points motivate us to pursue emerging research on transit-based ridesharing problem:
	    	
	    	\begin{enumerate}
                \item \textit{Complexity of a schedule-based (SB) transit service}: Previous studies such as \citet{Stiglic2018, Masoud2017a} have considered a cyclic (or frequency-based) transit service while designing a method for transit-based ridesharing. In the case of a cyclic timetable with identical trains, the travel time is not dependent on the departure time of the train from a transit station. However, large-scale transit systems incorporate traffic congestion and transit demand (to evaluate dwell time) while designing the schedule that usually varies during the day. Incorporating schedule becomes even more important when the transit service is not frequent, e.g., at park-and-ride stations, where buses arrive almost every hour. In that case, a frequency-based approach will result in an inaccurate estimation of wait time. The correct estimation of wait time, walking time, and transfer time are important to assess the feasibility of matches. We address this issue by proposing a systematic schedule-based ridesharing algorithm to find feasible matches. Furthermore, this would help in developing a transit-based ridesharing app that can provide directions in a multimodal transportation network.
	    		
	    		\item \textit{Drop-off station}: Previous studies have considered the closest stop from origin and destination for the drop-off and alighting (respectively) of a passenger. However, the closest stop may not give access to a faster transit route. Our proposed algorithm considers all the possible boarding and alighting stops. This will improve the possibility of assigning a faster service to the passenger. 
	    		
                \item \textit{Storage complexity}: Previous studies have used stored shortest transit travel itineraries. However, this is not possible for an SB transit network due to huge storage cost. For example, the Twin Cities transit system has 13673 stops in the network and for 1440 possible departure times in a day, we need to store  $2.1e^{12}$ possible travel times and itineraries. Furthermore, the change in the schedule for various days (weekdays, weekends, or holidays), or disruptions in the service makes it cumbersome to store such information. One of the contributions of this study is to avoid such storage with real-time transit direction computation with a fast algorithm.

	    	\end{enumerate}	    	
	    All the above-cited studies have made important contributions to this challenging problem, however, neither of the studies has considered the complexities of the transit part of the trip and showed application on a large-scale transit network. The current study leverages the schedule-based transit shortest path (SBTSP) algorithm to find potential matches between riders and drivers. The proposed algorithm uses the concept of space-time prism of both rider and driver to reduce the size of the network search. To solve this problem in real-time, a practical algorithm using a rolling horizon approach is proposed in this paper.  The proposed method not only matches riders and drivers for the first/last mile of the rider’s trip but also assigns an optimal transit station and a transit vehicle trip from the schedule for the rest of the rider's itinerary.

	 \section{Preliminaries}\label{sec:prelim}
	 	In this section, we discuss operation of a transit-based ridesharing program and describe notations and concepts useful to understand the matching problem. The ridesharing program receives a set of requests ($P = R \cup D$) that can be partitioned into a set of riders $R$, who are looking for a ride and have first mile/last mile problem, and a set of drivers $D$, who are willing to give a ride to the ride seekers to a transit station. Each request $p \in P$ is defined by a tuple $\{OR(p), DS(p), \tau^{at}(p), \tau^{pd}(p), \tau^{pa}(p), \delta(p), \Delta(p)\}$, indicating its origin $OR(p)$, destination $DS(p)$, announcement time $\tau^{at}(p)$, preferred departure time from origin $\tau^{pd}(p)$, preferred arrival time at destination $\tau^{pa}(p)$, maximum acceptable schedule deviation $\delta(p)$, and total travel delay $\Delta(p)$ allowed in the trip. Let $t_{ij}$  be the travel time from location $i$ to location $j$. Using above information, the earliest ($\tau^{ed}(p)$) and latest  ($\tau^{ld}(p)$) departure time from origin, earliest ($\tau^{ea}(p)$) and latest ($\tau^{la}(p)$) arrival time at destination, and maximum ride time ($t_{max}(p)$) can be calculated as:
	 	\begin{subequations}	
	 		\begin{align}		 
				 \tau^{ed}(p) = \tau^{pd}(p) - \delta(p)\\
				 \tau^{ld}(p) = \tau^{pd}(p) + \delta(p)\\
				 \tau^{ea}(p) = \tau^{pa}(p) - \delta(p)\\
				 \tau^{la}(p) = \tau^{pa}(p) + \delta(p)\\
				 t_{max} = t_{OR(p), DS(p)} + \Delta
			\end{align}
		 \end{subequations}	 
    In this ridesharing program, there is an online platform where riders and drivers can submit their requests, and the system would match these riders and drivers to travel together. A rider is matched with some driver who can drop off the rider at a transit station, from where the rider can take a bus or train to their destination. To get an insight into the problem, let us consider the following instance:\\

	\begin{figure}[h!]        
		\centering
		\begin{tikzpicture}[]    
		\node[shape=circle,draw=black] (1) at (0,0) {$OR(j)$};
		\node[shape=circle,draw=black] (2) at (8, 0) {$DS(j)$};
		\node[shape=circle,draw=black] (3) at (0,-5) {$OR(i)$};
		\node[shape=circle,draw=black] (4) at (8,-5) {$S(ij)$};
		\node[shape=circle,draw=black] (5) at (13, -5) {$DS(i)$};
		\path [->] (1) edge node[left, yshift=7pt, xshift =13pt] {$t_{OR(j)DS(j)} = 5$} (2);
		\path [->] (1) edge node[above, rotate=90, yshift=0pt, xshift =3pt] {$t_{OR(j)OR(i)} = 2$} (3);
		\path [->] (3) edge node[left, yshift=7pt, xshift =13pt] {$t_{OR(i)s} = 3$} (4);
		\draw[bend left, dashed, ->]  (4) to node [above,  yshift=-5pt, xshift =0pt] {$t_{S(ij)DS(i)} = 5$} (5);
		\draw[bend right,->]  (4) to node [below, , yshift=2pt, xshift =10pt] {$t_{S(ij)DS(i)} = 3$} (5);
		\path [->] (4) edge node[left, rotate=90, yshift=7pt, xshift =35pt] {$t_{S(ij)DS(j)} = 2$} (2);
		\path [->] (5) edge node[right, yshift=8pt, rotate=-45, xshift =-25pt] {$t_{DS(i)DS(j)} = 2$} (2);
		
		\end{tikzpicture}
		\begin{tikzpicture}[thick, scale=0.7, every node/.style={scale=0.6}]
		\draw[thick, dashed, ->] (-0.3,-4) -- (0.4,-4);    
		\node[text width=4cm] at (2.5,-4) {Transit link};
		\draw[thick, ->] (-0.3,-5) -- (0.4,-5);    
		\node[text width=4cm] at (2.5,-5) {Auto link};
		\end{tikzpicture}
		\caption{An example of transit-based ridesharing}
	\end{figure}
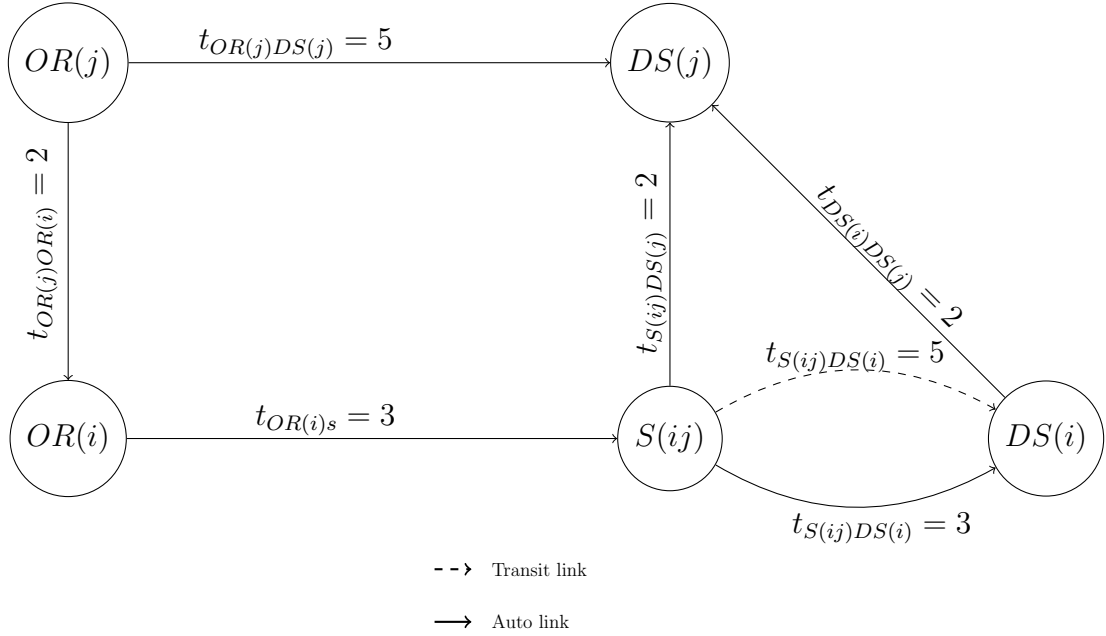

    A driver $j \in D$ is going from $OR(j)$ to $DS(j)$, which takes about 5 minutes by driving. A rider $i \in R$, going from $OR(i)$ to $DS(i)$ is facing the first-mile access problem and cannot take transit, so $i$ must drive to the destination, which takes 6 minutes of driving time. However, if $j$ agrees to give a ride to $i$, then $j$ can drop $i$ off at a station $S(ij)$, from where $i$ can take transit to her destination and $j$ can drive to her destination. The overall auto vehicle-minutes spent on the combined trip with transit-based ridesharing is 7 minutes in comparison to 5 + 6 = 11 minutes without ridesharing and 10 minutes with only ridesharing (i.e., driver $j$ going to rider's destination $DS(i)$ and then going to her destination $DS(j)$). Due to time constraints, it might not be possible for a driver to give a ride all the way to the rider's destination but can only give a ride to a transit station. This example shows that transit-based ridesharing can save a lot more vehicle-min than any other mode considered here and possibly more matching rate than a stand-alone ridesharing program.
	
	\subsection{Network topology}
	In this section, we describe the notations related to graphs characterizing the road and transit network as these notations will be used throughout the text. 
	
	\subsubsection{Road network}
    Let $G_R(N_R, A_R)$ be a digraph representing road network, where $N_R$ is the set of nodes and $A_R$ is the set of links connecting these nodes. We consider a static or non-timed road network for this study. Let $t: N_R \times N_R \rightarrow \mathfrak{R}_{+}$ be a function that returns the shortest travel time between two nodes in the road network. The efficient shortest path algorithms (such as Dijkstra's algorithm using binary heap data structure) take a fraction of a second to compute the shortest path on a network with thousands of nodes and links. For this network, we store the shortest-path trees in the memory (storage cost $\mathcal{O}(N_R^2)$) and call it whenever required by the matching algorithm. Let us denote the set of origins and the destinations of the participants as $O = \cup_{p \in P} OR(p)$ and $D = \cup_{p \in P} DS(p)$  respectively. Note that $O \cup D \subseteq N_R$. \\

	\subsubsection{Transit network}\label{sec:tn}
     A schedule-based (SB) transit network can capture the complexities of the movement of a passenger such as precise waiting time, in-vehicle time, and transfer time to other routes. This network $G(N_T, A_T)$ is a digraph that is created using service schedule data obtained from transit agencies. The General Transit Feed Specification (GTFS) is a standard format of transit schedule data released publicly by various transit agencies throughout the world. We use trip-based network representation for modeling the transit network (\cite{Khani2014}). In this procedure, the main consideration is given to the fact that every deviation from a transit route could not be considered as a transfer. Depending upon the acceptable walking time, waiting time, and direction of movement, transfers are created between two transit routes. The details about the creation of this network are described in the next paragraph. \\
	
	In the schedule data, let us denote the set of transit stops\footnote{A transit stop is a geographic location from where a traveler can board/alight the bus.} as $B$, set of transit routes/lines\footnote{A transit line is defined by a set of stops with a starting and an end point between which buses run back and forth.} as $L$, and set of trips\footnote{A trip is a travel itinerary of a bus with arrival and departure time specified at different stops.} as $K$. Each trip $k \in K$ is characterized by a route $l_k \in L$, set of stops (or nodes) $B_k \subset B\times K$, scheduled arrival/departure time $\tau: B_k \mapsto \mathfrak{R}_{+}$ at these stops, and a sequence $\zeta : B_k \mapsto \mathfrak{N}$ in which these stops are visited. The set of nodes in the transit network are defined as $N_T = \cup_{k \in K} B_k $. The set of links in the transit network, $A_T = A_{T_a} \cup A_{T_v} \cup A_{T_w}$ consists of three types of links, namely, acess/egress links $A_{T_a}$, in-vehicle links $A_{T_v}$, and walking/waiting transfer links $A_{T_w}$. Let $w: (N_T \cup N_R) \times (N_T \cup N_R) \mapsto \mathfrak{R}_{+}$ be the walking distance between two nodes. The access/egress arcs are walking links created between origins and transit stops or between transit stops and destinations if the distance between a pair is less than an acceptable walking distance $w_0$ (say 0.75 mi), i.e., $A_{T_a} = \{(i, j) : w(i, j) \le w_0 \text{ for some } i \in O, j \in N_T \} \cup \{(i, j) : w(i, j) \le w_0 \text{ for some } i \in N_T, j \in D \}$. The in-vehicle links $A_{T_v} = \{(i, j) : i, j \in B_k \text{ for some } k \in K\}$ are created using itinerary of any transit trip. Finally, the waiting/walking transfer links are defined as $A_{T_w} = \{(i, j) 
	\ | \ i \in B_k, j \in B_{k^{'}} \text{ for some } k, k' \in K, l_i \ne l_{j}, \zeta(i) \ne 1, \zeta(j) \ne \max_{m \in B_{k^{'}}} \zeta(m), w(i, j ) \le w_1, \vert \tau(j) - \tau(i) - w(i, j) \vert \le \rho\}$, where $\rho$ is the threshold by which a bus trip can be early or late, and $w_1$ is the acceptable walking distance (say 0.25 mi) for transfers. \\
	
	Finally, let $Z:N_T \mapsto N_R$ be a function which outputs closest road node for any transit node, i.e., $Z(s) = \text{argmin} \{ w(n, s) : n \in N_R\}$. Let $A_m = \{(i, j): i = Z(j) \text{ for some } j \in N_T \}$ be the mode transfer links which are used to access the transit network from a road node and vice-versa. The overall multimodal network is denoted by $G(N, A)$, where $N = N_R \cup N_T$, and $A = A_R \cup A_T \cup A_m$. \\
	
	\subsection{Space-time prism}
    Before delving into the matching problem, let us introduce a few more definitions which will help create a network search reduction technique later in the \cref{sec:sbtsp}. The following definitions are given in the context of ridesharing, however, more general definitions can be found in \citet{Miller2017}.\\

	Because of spatial and temporal limitations, a traveler can only be present in one place at a time due to which the activities performed by a traveler in space are restricted by the available time budget. Space-time prism provides us a framework that recognizes underlying constraints on human activities in space-time and also provides us an effective way of keeping track of these conditions.
	
	\begin{defn}(Anchor) An anchor is a node in the network at which a traveler begins/ends her journey. For a participant $p \in P$ in the ridesharing program, $OR(p)$ and $DS(p)$ are possible anchors.         
	\end{defn}

	\begin{defn}(Time-forward and time-backward cone) The time-forward cone, $FC_p\left(\tau\right)$, for a participant $p \in P$  pointed at the anchor $OR(p)$ is the set of nodes that can be reached at time $\tau$ within a given time budget. Conversely, the time-backward cone, $BC_p\left(\tau\right)$, for a participant $p \in P$ pointed at the anchor $DS(p)$  is a set of nodes from where $DS(p)$ can be reached within the remaining time budget $\tau^{la}(p) - \tau$.    Mathematically,
		\begin{subequations}    
			\begin{align}
			FC_p\left(\tau\right)= & \{ n \ : \ \tau \ge \tau^{ed}(p)  + t_{OR(p)n}, \tau \le \tau^{la}(p) \} \quad   \label{eq:1}\\
			BC_p\left(\tau\right)= & \{n \ : \ \tau \le \tau^{la}(p) - t_{nDS(p)}, \tau \ge \tau^{ed}(p) \}\quad   \label{eq:2}\
			\end{align}
		\end{subequations}    
	\end{defn}
	
	\begin{defn}(Space-time prism) The space time prism, $STP(p)$, for a participant $p \in P$ is defined as
		\begin{equation}\label{eq:82}
		STP_p(\tau) = FC_p\left(\tau\right) \cap BC_p\left(\tau\right) 
		\end{equation}
	\end{defn}

	\begin{defn}(Potential path area) The space that is accessible to a participant within a given time budget is known as potential path area (PPA).
		\begin{equation}\label{eq:3}
		{PPA}_p=  \{ n : t_{OR\left(p\right)n}+\ t_{nDS\left(p\right)}\le \tau^{la}\left(p\right)- \tau^{ed}\left(p\right) \}  
		\end{equation}
	\end{defn}
	
	The anchors, time-forward cone, time-backward cone, potential path area, and space-time prism are visualized in Figure \ref{fig:stp}. The constraints defining the space-time prism can be used to reduce the network search for an optimal itinerary. 
	
	\begin{prop}\label{prop:0}
		(Potential Path Area) The projection of space-time prism onto space is the potential path area $PPA_p$ for participant $p$.
	\end{prop}
	\begin{proof}
		Using (\ref{eq:1}) and (\ref{eq:2}), we have
		\begin{align}
		-\tau \le -\tau^{ed}(p)  - t_{OR(p)n} \label{eq:4}\\
		\tau \le \tau^{la}(p) - t_{nDS(p)} \label{eq:5}
		\end{align}
		Adding (\ref{eq:4}) and (\ref{eq:5}) yields,
		\begin{align}
		&\tau^{la}(p) -\tau^{ed}(p)  - t_{nDS(p)} - t_{OR(p)n} \ge 0 \nonumber\\
		\implies & \tau^{la}(p) -\tau^{ed}(p) \ge t_{OR(p)n} + t_{nDS(p)} \label{eq:6}.
		\end{align}
		which is same as \eqref{eq:3}.
	\end{proof}

	\begin{figure}[h!]
		\centering
		\begin{tikzpicture}
		\draw[->, thick] (0,0)--(8,0) node[right]{};
		\draw[->, thick] (0,0)--(0,8) node[right]{};
		\draw[-, thick] (0,0)--(2, 2) node[right]{};
		\draw[-, thick] (2,2)--(8,2) node[right]{};
		\draw[-, thick] (8,2)--(6,0) node[right]{};
		\draw (4,1) ellipse (1.5cm and 0.25cm);
		\draw[-, dashed] (2.5,1)--(2.5,7) node[right]{};
		\draw[-, dashed] (5.5,1)--(5.5,7) node[right]{};
		\draw[-, dashed] (2.5, 4.7)--(1.5, 5.5) node[right]{};
		\draw[-, dashed] (2.5, 4.7)--(1.5, 4) node[right]{};
		\draw[-, dashed] (5.5, 5.3) --(6.5, 6.2) node[right]{};
		\draw[-, dashed] (5.5, 5.3) --(6.5, 4.4) node[right]{};        
		\draw[->] (6, 7)--(4, 5.5) node[right]{};
		\node[font=\fontsize{8}{10}] at (6.2, 7.2) {$BC_p(\tau)$};        
		\draw[->] (2, 3)--(3.5, 4.5) node[right]{};
		\node[font=\fontsize{8}{10}] at (1.8, 2.8) {$FC_p(\tau)$};        
		\draw[->] (6, 3)--(4, 5) node[right]{};
		\node[font=\fontsize{8}{10}] at (6.2, 2.8) {$CY_p(\tau)$};        
		\draw[rotate = 14.5] (5.1,3.8) ellipse (1.5cm and 0.15cm);        
		\draw[fill = black] (3.5,1) ellipse (0.05cm and 0.01cm);
		\draw[fill = black] (4.5,1) ellipse (0.05cm and 0.01cm);
		\draw[-, dashed] (3.5, 4)--(3.5,1) node[right]{};
		\draw[-, dashed] (4.5, 4)--(4.5,1) node[right]{};
		\draw[-, dashed] (3.5, 4)--(0,4) node[right]{};
		\draw[-, dashed] (4.5, 6.2)--(0,6.2) node[right]{};        
		\draw[-, thick] (3.5, 4)--(2.5, 4.7) node[right]{};
		\draw[-, thick] (3.5, 4)--(5.5, 5.3) node[right]{};
		\draw[-, thick](2.5, 4.7)--(4.5, 6.2) node[right]{};
		\draw[-, thick](4.5, 6.2)--(5.5, 5.3) node[right]{};
		\draw[-, dashed] (4.5, 8)--(4.5,1) node[right]{};
		\draw[-, thick] (4.5, 7)--(4.5,6.2) node[right]{};
		\draw[-, thick] (3.5, 4)--(3.5, 3) node[right]{};    
		\node[text width=4cm] at (1.5, 8.2) {$Time$};
		\node[text width=4cm] at (9.4, -0.3) {$Space$};
		\node[font=\fontsize{8}{10}] at (3.5,0.9) {$OR(p)$};
		\draw[<-] (3.5,1)--(-1.5,3) node[right]{};
		\draw[<-] (4.5,1)--(-1.5,3) node[right]{};
		\node[font=\fontsize{5}{10}] at (-1.5,3.2) {ST anchors};
		\node[font=\fontsize{8}{10}] at (4.5,0.9) {$DS(p)$};
		\node[font=\fontsize{10}{10}] at (-0.5,4) {$\tau^{ed}(p)$};
		\node[font=\fontsize{10}{10}] at (-0.5,6.2) {$\tau^{la}(p)$};    
		\node[font=\fontsize{7}{10}] at (3.8,0.5) {$PPA_p$};    
		\draw[->, thick] (0,0)--(0,8) node[right]{};    
		\end{tikzpicture}
		
		\caption{Space-time prism ($STP(p)$) for participant $p$}
		\label{fig:stp}
	\end{figure}
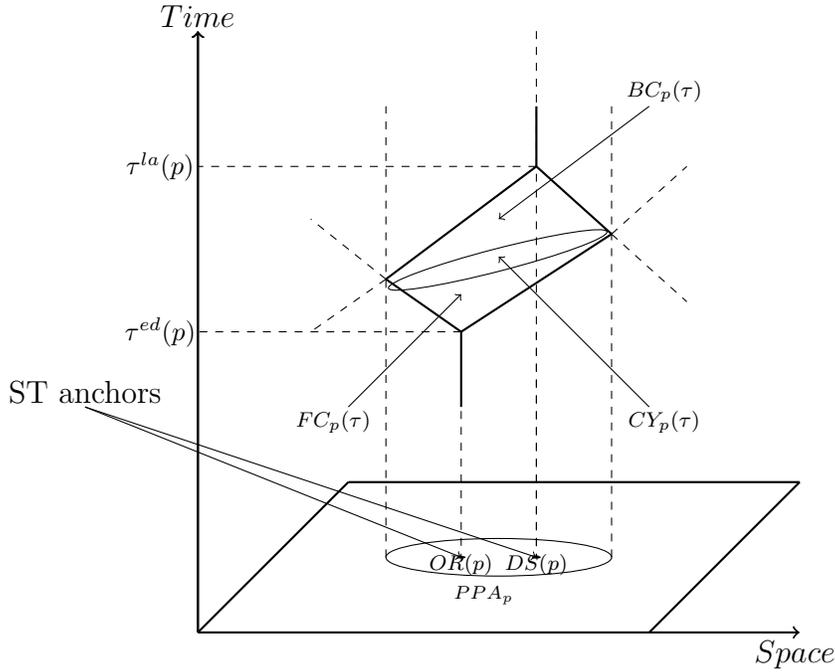
	
	\section{Transit-based ridesharing matching problem}\label{sec:tbrp}
	 In this section, we define the transit-based ridesharing problem. We make several assumptions that help in defining our problem. 
	\subsection{Assumptions}\label{sec:assp}
	\begin{enumerate}
	    \item All buses and trains are assumed to have sufficient capacity. Passengers do not face denied boarding due to capacity constraints.
		\item Transit service is reliable, which means that all buses and trains arrive on time according to their schedule.  
		\item Time spent in boarding and alighting a transit vehicle or getting on or off a car is negligible and can be incorporated as a service time ($t^{ser}$). The service time can also be used to account for delays in the schedule. 
		\item Transfer of a passenger between several drivers (multi-hop) is not allowed as this can be inconvenient for a rider and reduce the attractiveness of the program.
	\end{enumerate}
	\subsection{Problem Statement}
	The transit-based ridesharing matching problem is defined as follows. Given a set of requests $P$, a road network $G_R$, and a transit network $G_T$, find the following:
		\begin{enumerate}
			\item An optimal match between drivers and riders for ridesharing so that a given objective function is optimized, and
			\item Optimal drop off transit stops location $s \in B$ for the riders, and
			\item Optimal transit itinerary $\mathcal{I}$ for the riders.
		\end{enumerate}

	\begin{defn}(Feasible match) A match ($r, d, s, \mathcal{I}$) between a rider $r$ and driver $d$, drop off stop location $s$, and transit trip itinerary $\mathcal{I}$ is feasible if and only if it satisfies the following constraints:
		\begin{enumerate}
			\item $d$ should depart from $OR(d)$ after $\tau^{ed}(d)$
			\item $d$ should reach at $OR(r)$ after $\tau^{ed}(r)$ 
			\item $r$ should reach at $DS(r)$ before $\tau^{la}(r)$ 
			\item $d$ should reach at $DS(d)$ before $\tau^{la}(d)$ 
			\item $\mathcal{I}$ should be optimal for $r$ i.e., best possible transit path from stop $s$ consisting of different weights to transit trip components
		\end{enumerate}
	\end{defn}

    A feasible match between a driver and a rider must satisfy time constraints associated with the earliest departure and latest arrival time. They restrict passenger access to each node in the network at different times. The space-time prism in \eqref{eq:82} captures above four conditions. By making use of it, we make the following proposition. 
		\begin{prop}\label{prop:1}
		A driver $d \in D$ and a rider $r \in R$ can share a ride together if and only if $OR(r) \in STP_r(\tau) \cap STP_d(\tau)$ as well as $Z(s_k) \in  STP_r(\tau) \cap STP_d(\tau)$, where $s_k \in N_T$ is the drop-off node of $r$.
	\end{prop}
	\begin{proof}
		($\implies$) We prove this by contradiction. Let $J = \{OR(r), Z(s_k)\}$. Let us assume that $J \subsetneq STP_d(\tau)$, then for some $n \in J$, $n \notin PPA_d$, i.e.,  $\tau^{ed}(d) + t_{OR(d) n} + t_{n DS(d)} > \tau^{la}(d)$, which means $d$ cannot reach to her destination in time, thus violating the time constraints which is a contradiction according to (\ref{eq:1})-(\ref{eq:3}). Furthermore, $OR(r)$ is clearly in $STP(r)$ as this node is the anchor of time-forward cone. However, if $Z(s_k) \notin STP_r(\tau)$, then, $Z(s_k) \notin PPA_r$, i.e., $\tau^{ed}(r) + t_{OR(r) Z(s_k)} + t_{s_k DS(r)} > \tau^{la}(r)$, resulting in the violation of time constraints for a feasible match.
		
		\noindent ($\impliedby$) If $J \subseteq  STP_r(\tau) \cap STP_d(\tau)$, then both nodes present in $J$ are accessible by $r$ and $d$ without violating the time budget constraints for any of them, which means they can be present at these nodes at the same time, making it possible to share a ride. Therefore, the proposition holds.
	\end{proof}
			\begin{figure}[h!]
		\centering
		\begin{tikzpicture}
		\draw[->, thick] (0,0)--(10,0) node[right]{}; %
		\draw[->, thick] (0,0)--(0,8) node[right]{};
		\draw[-, thick] (0,0)--(2, 2) node[right]{};
		\draw[-, thick] (2,2)--(8,2) node[right]{};
		\draw[-, thick] (8,2)--(6,0) node[right]{};
		\draw (4,1) ellipse (1.5cm and 0.4cm);
		\draw[-, dashed] (2.5,1)--(2.5,7) node[right]{};
		\draw[-, dashed] (5.5,1)--(5.5,7) node[right]{};
		\draw[-, dashed] (2.5, 4.7)--(1.5, 5.5) node[right]{};
		\draw[-, dashed] (2.5, 4.7)--(1.5, 4) node[right]{};
		\draw[-, dashed] (5.5, 5.3) --(6.5, 6.2) node[right]{};
		\draw[-, dashed] (5.5, 5.3) --(6.5, 4.4) node[right]{};

		\draw[fill = black] (3.5,1) ellipse (0.05cm and 0.01cm);
		\draw[fill = black] (4.5,1) ellipse (0.05cm and 0.01cm);
		\draw[fill = black] (4.8,1.15) ellipse (0.05cm and 0.01cm);
		\draw[-, dashed] (3.5, 4)--(3.5,1) node[right]{};
		\draw[-, dashed] (4.5, 4)--(4.5,1) node[right]{};
		\draw[-, dashed] (3.5, 4)--(0,4) node[right]{};
		\draw[-, dashed] (4.5, 6.2)--(0,6.2) node[right]{};
		
		\draw[-, thick] (3.5, 4)--(2.5, 4.7) node[right]{};
		\draw[-, thick] (3.5, 4)--(5.5, 5.3) node[right]{};
		\draw[-, thick](2.5, 4.7)--(4.5, 6.2) node[right]{};
		\draw[-, thick](4.5, 6.2)--(5.5, 5.3) node[right]{};
		\draw[-, dashed] (4.5, 8)--(4.5,1) node[right]{};
		\draw[-, thick] (4.5, 7)--(4.5,6.2) node[right]{};
		\draw[-, thick] (3.5, 4)--(3.5, 3) node[right]{};

		\node[text width=4cm] at (1.5, 8.2) {$Time$};
		\node[text width=4cm] at (12.1, 0) {$Space$};
		\node[font=\fontsize{2}{5}] at (3.5,0.9) {$OR(d)$};
		\node[font=\fontsize{2}{5}] at (4.5,0.8) {$DS(d)$};
		\node[font=\fontsize{8}{10}] at (-0.4,3.9) {$\tau^{ed}(d)$};
		\node[font=\fontsize{8}{10}] at (-0.4,6.1) {$\tau^{la}(d)$};

		\draw[->, thick] (0,0)--(0,8) node[right]{};

		\draw[->] (2, 7)--(4, 5.5) node[right]{};
		\node[font=\fontsize{8}{10}] at (2, 7.1) {$STP(d)$};
		
		\draw[->] (7, 4)--(5.6, 5.7) node[right]{};
		\node[font=\fontsize{8}{10}] at (7.1, 3.8) {$STP(r)$};

		\begin{scope}[shift={(1,0.2)}]
		\draw (4,1) ellipse (1.5cm and 0.25cm);
		
		\draw[-, dashed] (2.5,1)--(2.5,7) node[right]{};
		\draw[-, dashed] (5.5,1)--(5.5,7) node[right]{};
		\draw[-, dashed] (2.5, 4.7)--(1.5, 5.5) node[right]{};
		\draw[-, dashed] (2.5, 4.7)--(1.5, 4) node[right]{};
		\draw[-, dashed] (5.5, 5.3) --(6.5, 6.2) node[right]{};
		\draw[-, dashed] (5.5, 5.3) --(6.5, 4.4) node[right]{};

		\draw[fill = black] (3.5,1) ellipse (0.05cm and 0.01cm);
		\draw[fill = black] (4.5,1) ellipse (0.05cm and 0.01cm);
		\draw[-, dashed] (3.5, 4)--(3.5,1) node[right]{};
		\draw[-, dashed] (4.5, 4)--(4.5,1) node[right]{};
		\draw[-, dashed] (3.5, 4)--(-1,4) node[right]{};
		\draw[-, dashed] (4.5, 6.2)--(-1,6.2) node[right]{};
		
		\draw[-, thick] (3.5, 4)--(2.5, 4.7) node[right]{};
		\draw[-, thick] (3.5, 4)--(5.5, 5.3) node[right]{};
		\draw[-, thick](2.5, 4.7)--(4.5, 6.2) node[right]{};
		\draw[-, thick](4.5, 6.2)--(5.5, 5.3) node[right]{};
		\draw[-, dashed] (4.5, 8)--(4.5,1) node[right]{};
		\draw[-, thick] (4.5, 7)--(4.5,6.2) node[right]{};
		\draw[-, thick] (3.5, 4)--(3.5, 3) node[right]{};

		\node[font=\fontsize{5}{10}] at (3.1,1) {$OR(r)$};
		\node[font=\fontsize{5}{10}] at (4.1,0.9) {$Z(s)$};
		\node[font=\fontsize{5}{10}] at (5,1) {$DS(r)$};
		\node[font=\fontsize{8}{10}] at (-1.4,4.1) {$\tau^{ed}(r)$};
		\node[font=\fontsize{8}{10}] at (-1.4,6.3) {$\tau^{la}(r)$};
		\end{scope}
		\end{tikzpicture}		
		\caption{Feasible region where $r$ and $d$ can meet}
		\label{fig:int}	
	\end{figure}
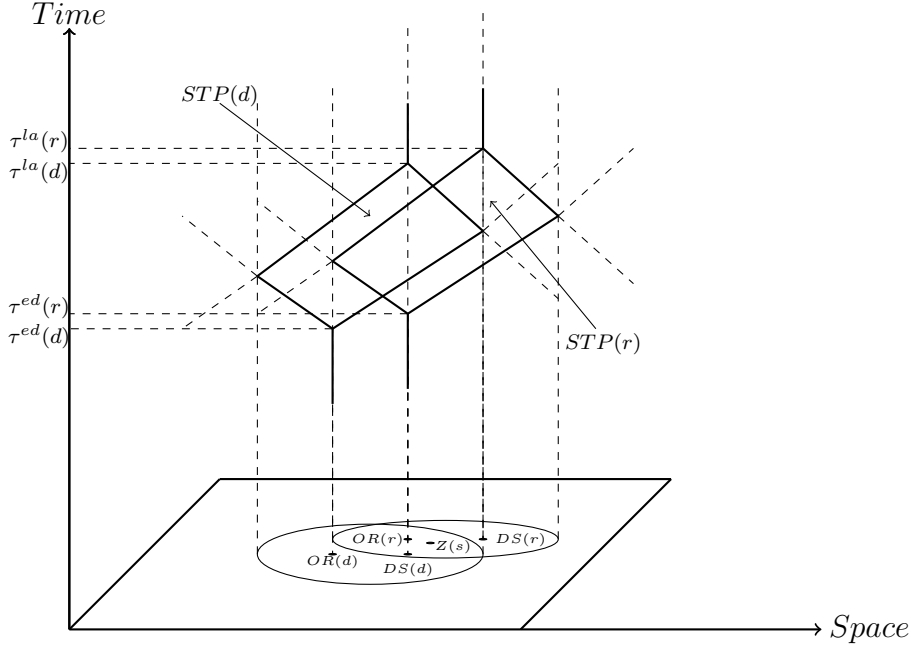
	
    Following Proposition \eqref{prop:1}, $STP_r(\tau) \cap STP_d(\tau)$ provides us the set of potential nodes in the network where $r$, and $d$ can be present at the same time $\tau$. These nodes can be used for potential rideshare between $r$ and $d$ as shown in Figure \ref{fig:int}. This approach will greatly reduce the size of the network to be considered for ridesharing and thus reduce the computational time for finding a feasible match between a rider and driver.\\
	
	\section{Solving transit-based ridesharing matching problem}\label{sec:sbtsp}
	
    For simplicity, we present a methodology for riders facing the first-mile access problem, i.e., no transit access at the origin for inbound trips due to limited transit network coverage or infrequent service in suburbs but sufficient transit access to destinations in the city center. The procedure can be reversed for the last mile access problem (for outbound trips), a discussion on which is given in \cref{sec:lm}.  The process of finding optimal matches for transit-based ridesharing follows a two-step procedure. The first step finds a set of feasible matches between drivers and riders, and in the second step, we use a matching optimization program to assign riders to drivers. They are described in detail below:
	
	 \subsection{An algorithm to find feasible matches}
    The proposed method leverages \textit{schedule-based transit shortest path} (SBTSP) in developing an algorithm for finding feasible matches. The quickest path in an SB transit network may not be optimal for a passenger. For example, the quickest path may have several transfers that are less attractive to a passenger. By assigning different weights to the cost of different links, SBTSP finds an optimal path according to user route choice behavior. For example, waiting and transfers are considered onerous components of a transit trip and hence can be assigned a higher weight. The developed algorithm returns a set of feasible matches for the riders and drivers participating in this program. We assume that the reader is familiar with shortest path labeling algorithms for the discussion to follow in the paragraph below: \\

	Given that $r$ departs from node $n \in N_T$, we find the labels $\gamma_n$  that specify the latest time $r$ should depart from $n$ in order to reach $DS(r)$ before $\tau^{la}(r)$. Let $SEL$ be a scan eligible list, $\xi_i$ be the predecessor of node $i$ in shortest path, $\Gamma^{-1}(i) = \{j \ | \ (j, i) \in A_T \}$ be the backward star of node $i$ (i.e., collection of all adjacent nodes from where node $i$ can be reached). We now  define weights associated with different types of links in $G_T$ for calculating a generalized cost. Let $\eta_a, \eta_v, \text{and } \eta_w$ be the weights associated with links in set ${A_T}_a, {A_T}_v, \text{and } {A_T}_w$ respectively. Using weighted sum of the cost of traversing different types of link in the transit network, a generalized cost label $\gamma^{gc}_n$ is maintained for each node $n \in {N_T}$.\\
	
	\noindent \textit{Bellman's prinicipal of optimality}: For any node $i \in {N_T}, \gamma^{gc}_i$ should satisfy the following condition:
	\begin{equation}\label{eq:bpo}
	\gamma^{gc}_j = \underset{j \in \Gamma(i)}{\text{min}} \{\gamma^{gc}_i + t_{ij} \}
	\end{equation}
	Pseudocode for finding feasible matches is given in Algorithm~\ref{alg:sbtsp}. The algorithm starts with the initialization of a set of potential matches as $\mathcal{M} = \phi$. For every rider $r \in R$, we run a backward shortest path from $DS(r)$.  We initialize a scan eligible list (SEL) and maintain two different types of labels--time labels $\gamma$ and generalized cost labels $\gamma^{gc}$. The time labels are used to check if the time constraints are satisfied for both riders and drivers while the generalized cost labels $\gamma^{gc}$ are used to maintain a minimum generalized cost for a rider which is calculated as the weighted sum of the cost of traversing different types of link in the transit network (line 11-16). Line 18 checks the Bellman's principle of optimality (\ref{eq:bpo}) and restrict the network search by using driving time as a lower bound on transit time from $OR(r)$ to node $j$. While updating labels of each node $j \in {N_T}$, each driver $d \in D$ is checked for her compatibility with given rider $r$. Line 23 checks if $d$ can reach $OR(r)$ by $\tau^{ed}(r)$, then we further check if $d$ can reach $DS(d)$ before $\tau^{la}(d)$. In other words, these two \emph{if} statements check if $j \in STP_d(\gamma) \cap STP_r(\gamma)$. If all these conditions are satisfied, then we add $(r, d, j, \xi)$ to $\mathcal{M}$. The node $j$ and predecessor set $\xi$ can be used to retrieve the drop off station location $s$ and optimal itinerary $\mathcal{I}$ for $r$. The procedure will be reversed for the last mile problem, in which case, we run the algorithm in the forward direction.   
	
	\begin{algorithm}[H]
		\caption{Pseudocode for finding rider driver feasible matches}\label{euclid}
		\label{alg:sbtsp}
		\begin{algorithmic}[1]
			\Procedure{RiderDriverPotentialMatch}{$R, D, G, \eta$}
			\State $\mathcal{M} \gets$  \{\} \Comment{intializing set of potential matches}
			\For{$r$ in $R$}
			\State $i \gets DS(r), \gamma_i \gets \tau^{la}(r), \gamma^{gc}_i \gets 0, \xi_i = \text{NULL}$        
			\State $\gamma_j \gets \infty, \gamma^{gc}_j \gets \infty, \xi_j \gets \phi \quad \forall j  \ne i$     \Comment{intializing node labels and SEL}
			\State SEL $\gets \{i\}$
			\While{SEL $\ne \phi$}
			\State $i \gets \underset{k}{\text{argmin}} \{\gamma^{gc}_k \ | \ k \in \text{SEL}\}$
			\State SEL $\gets$ SEL$\backslash\{i\}$
			\For{each $j \in \Gamma^{-1}(i)$}
			\If{$(j, i) \in {A_T}_a$}
			\State $\gamma^{gc}_{new} \gets \gamma^{gc}_i + \eta_a*t_{ji}$
			\ElsIf{$(j, i) \in {A_T}_v$}
			\State $\gamma^{gc}_{new} \gets \gamma^{gc}_i + \eta_v*t_{ji}$
			\Else
			\State$\gamma^{gc}_{new} \gets \gamma^{gc}_i + \eta_w*t_{ji}$
			\EndIf
			\State$\gamma_{new} \gets \gamma_i -t_{ji}$
			\If{$ \gamma^{gc}_{new} < \gamma^{gc}_j$ and $\tau^{ed}(r) + t_{OR(r)Z(j)} + t^{ser} < \gamma_{new}$} 
			
			\State $\gamma^{gc}_j \gets \gamma^{gc}_{new}, \gamma_j \gets \gamma_{new}, \xi_j \gets i$
			\State SEL $\gets$ SEL $\cup \{j\}$ \Comment{Updating labels}

			\For{$d$ in $D$}
			\State $\tau_{OR(r)} \gets \tau^{ed}(d) + t_{OR(d)OR(r)}$
			\If{$\tau^{ed}(r) \le \tau_{OR(r)}  \le \tau^{ld}(r)$}
			\If{ $\tau_{OR(r)} + t_{OR(r)Z(j)} + t^{ser} \le \gamma_{j}$  \Comment{checking time constraints}
				\State $\text{and} \  \tau_{OR(r)} + t_{OR(r), Z(j)} + t^{ser} + t_{Z(j), DS(d)} \le\tau^{la}(d)$}
			
			\State $\mathcal{M} \gets \mathcal{M} \cup \{(r, d, j, \xi)\}$ 
			
			\EndIf
			\EndIf
			
			\EndFor
			\EndIf
			\EndFor
			
			\EndWhile 
			\EndFor    
			\Return $\mathcal{M}$
			\EndProcedure        
		\end{algorithmic}
	\end{algorithm}

	To illustrate the use of above algorithm, consider an example network given in \figurename{\ref{fig:trips}}. In this graph, 
	\begin{align*}
	{N_R} = & \{OR(r), OR(d), n_1, n_2, DS(r), DS(d)\}\\
	B = & \{s_1, s_2, s_3, s_4, s_5\}\\
	{N_T} = & \{u_1, u_2, u_3, u_4, u_5, u_6, u_1^{'}, u_2^{'}, u_3^{'}, u_4^{'},u_5^{'}, u_6^{'}\}    
	\end{align*}
	Assume $\tau^{ed}(d)= 0$ min, $\tau^{la}(d)= 20$ min, $\tau^{la}(r)= 40$ min, and $\tau^{ed}(r)= 10$ min.  Suppose the driver $d$ starts at $\tau^{ed}(d) = 0$ min and reach $OR(r)$ at 10 min, from where she has option to go to $n_1$ or $n_2$. However, going to $n_2$ and then $DS(d)$ will make the driver late as $\tau^{la}(d)= 20$ min. So $d$ can drop-off $r$ at $n_1$, from where rider can walk to station $s_1$. The service time in getting off and walking to the transit station is included in the access links. We assume weights to different types of transit links as  $\eta_a = 1, \eta_w = 2, \eta_v = 1$. Now there are three possible paths to go from $s_1$ to $DS(r)$, namely
	\begin{align*}
	\pi_1 = & \{(s_1, u_1), (u_1, u_1^{'}), (u_1^{'}, s_3), (s_3, u_6), (u_6, u_6^{'}), (u_6^{'}, s_4), (s_4, DS(r))\}\\
	\pi_2 = & \{(s_1, u_2), (u_2, u_2^{'}), (u_2^{'}, s_4), (s_4, DS(r))\}\\
	\pi_3 = & \{(s_1, u_3,), (u_3, u_3^{'}), (u_3^{'}, s_4), (s_4, DS(r))\}
	\end{align*}
	The time of arrival at $DS(r)$ using $\pi_1$, $\pi_2$, and $\pi_3$ are 36 min, 38 min, and 44 min. Therefore, $\pi_3$ is not feasible as time of arrival at $DS(r)$ is greater than $\tau^{la}(r) = 40$ min. We can see that $\pi_1$ is the quickest path to $DS(r)$. However, the generalized cost to reach $DS(r)$ using $\pi_1$ and $\pi_2$ are  43 min and 40 min respectively. The higher generalized cost of $\pi_1$ can be attributed to transfers and the waiting involved. In real transit networks, sometimes this can greatly affect user perception and should be incorporated in shortest path calculations. Moreover, a path involving more transfers may not be reliable and can result in more waiting time than scheduled. Therefore, the rider would prefer to take $\pi_2$ to reach her destination.
	
	\begin{figure}[h!]        
	\begin{subfigure}{\textheight}
		\begin{tikzpicture}[scale = 0.7]
		
		\node[scale=2] at (0, -1.5) {\faHome};
		\node[shape=circle,draw=black] (1) at (0,0) {\tiny $OR(d)$};
		\node[shape=circle,draw=black] (2) at (3,0) {\tiny $OR(r)$};
		\node[scale=2] at (3, -1.5) {\faHome};
		\node[shape=circle,draw=black] (3) at (6,3) {\small $n_1$};
		\node[shape=circle,draw=black] (4) at (6,-3) {\small $n_2$};
		\node[scale=1] at (8,4) {\faTrain};
		\node[shape=circle,draw=black] (5) at (8, 3) {\small $s_1$};
		\node[shape=circle,draw=black] (6) at (8,-3) {\small $s_2$};
		\node[scale=1] at (8,-4) {\faTrain};
		\node[shape=circle,draw=black] (7) at (11, 5) {\small $u_1$};
		\node[shape=circle,draw=black] (8) at (15, 5) {\small $u_1^{'}$};
		\node[shape=circle,draw=black] (9) at (13, 8) {\small $s_3$};
		\node[scale=1] at (12, 8) {\faTrain};
		\node[shape=circle,draw=black] (10) at (16, 9) {\small $u_6$};
		\node[shape=circle,draw=black] (11) at (19, 7) {\small $u_6^{'}$};
		\node[shape=circle,draw=black] (12) at (11, 3) {\small $u_2$};
		\node[shape=circle,draw=black] (13) at (17, 3) {\small $u_2^{'}$};
		\node[shape=circle,draw=black] (14) at (11, 1) {\small $u_3$};
		\node[shape=circle,draw=black] (15) at (17, 1) {\small $u_3^{'}$};
		\node[shape=circle,draw=black] (16) at (11,-2) {\small $u_4$};
		\node[shape=circle,draw=black] (17) at (17,-2) {\small $u_4^{'}$};
		\node[shape=circle,draw=black] (18) at (11,-4) {\small $u_5$};
		\node[shape=circle,draw=black] (19) at (17,-4) {\small $u_5^{'}$};
		\node[shape=circle,draw=black] (20) at (21, 4) {\small $s_4$};
		\node[scale=1] at (22, 4) {\faTrain};
		\node[shape=circle,draw=black] (21) at (21, -3) {\small $s_5$};
		\node[scale=1] at (22, -3) {\faTrain};
		\node[shape=circle,draw=black] (24) at (23, 0) {\tiny $DS(r)$};
		\node[scale=2] at (24.5, 0) {\faBuildingO};
		\node[shape=circle,draw=black] (25) at (8, 0) {\tiny $DS(d)$};
		\node[scale=2] at (9.5, 0) {\faBuildingO};
		\path [->] (1) edge node[left, yshift=3pt, xshift =7pt] {\tiny 10} (2);
		\path [->] (2) edge node[left, yshift=3pt, xshift =7pt] {\tiny 6} (25);
		\path [->] (2) edge node[left, yshift=2pt, xshift =4pt] {\tiny 5} (3);
		\path [->] (2) edge node[left, yshift=3pt, xshift =7pt] {\tiny 6} (4);
		\path [->] (4) edge node[left, yshift=2pt, xshift =4pt] {\tiny 7} (25);
		\path [->] (3) edge node[left, yshift=3pt, xshift =7pt] {\tiny 5} (25);
		\path [->, dashed] (3) edge node[left, yshift=3pt, xshift =5pt] {\tiny 2} (5);
		\path [->, dashed] (4) edge (6);
		\path [->, dotted] (5) edge node[left, yshift=3pt, xshift =5pt] {\tiny 4} (7);
		\path [->, dotted] (5) edge node[left, yshift=3pt, xshift =7pt] {\tiny 2} (12);
		\path [->, dotted] (5) edge node[left, yshift=3pt, xshift =7pt] {\tiny 2} (14);
		\path [->, dotted] (6) edge (16);
		\path [->, dotted] (6) edge (18);
		\path [->, loosely dashdotted] (7) edge node[left, yshift=3pt, xshift =7pt] {\tiny 3} (8);
		\path [->, densely dashdotted] (8) edge node[left, yshift=3pt, xshift =7pt] {\tiny 2} (9);
		\path [->, dotted] (9) edge node[left, yshift=3pt, xshift =7pt] {\tiny 1} (10);
		\path [->, loosely dashdotted] (10) edge node[left, yshift=3pt, xshift =7pt] {\tiny 4} (11);
		\path [->, loosely dashdotted] (12) edge node[left, yshift=3pt, xshift =7pt] {\tiny 12} (13);
		\path [->, loosely dashdotted] (14) edge node[left, yshift=3pt, xshift =7pt] {\tiny 15} (15);
		\path [->, loosely dashdotted] (16) edge (17);
		\path [->, loosely dashdotted] (18) edge(19);
		\path [->, dashed] (11) edge node[left, yshift=3pt, xshift =7pt] {\tiny 2} (20);
		\path [->, dashed] (13) edge node[left, yshift=3pt, xshift =7pt] {\tiny 4} (20);
		\path [->, dashed] (15) edge node[left, yshift=3pt, xshift =7pt] {\tiny 7} (20);
		\path [->, dashed] (17) edge (21);
		\path [->, dashed] (19) edge (21);
		\path [->, dashed] (20) edge node[left, yshift=3pt, xshift =7pt] {\tiny 3} (24);
		\path [->, dashed] (21) edge (24);
		\draw[-, line width=0.25mm,blue] (0, 0.8) -- (8, 5);
		\draw[-, line width=0.25mm,blue] (0, -0.8) -- (8, -2);
		\draw[-, line width=0.25mm,blue] (8, -0.8) -- (3, -3);
		\draw[-, line width=0.25mm,blue] (8, 0.8) -- (6, 5);
		\draw[-, line width=0.25mm, green] (3.5, -0.5) -- (18,-4);
		\draw[-, line width=0.25mm,green] (21, 3.5) -- (6, -2);
		\draw[-, line width=0.25mm, green] (21, 4.5) -- (16, 13.5);
		\draw[-,line width=0.25mm, green] (2.2, 0.4) -- (18, 13.5);
		\node[font=\fontsize{8}{10}, blue] at (3, 3) {$STP(d)$};
		\node[font=\fontsize{8}{10}, green] at (17.5, 12.1) {$STP(r)$};
		
		\end{tikzpicture}
	\end{subfigure}
	
	\begin{subfigure}{\textheight}
		\begin{tikzpicture}[scale=0.6, every node/.style={scale=0.6}]
		\draw[->] (0,0) -- (2,0);    
		\node[text width=4cm] at (4,0) {Road link};
		\draw[->, dashed] (5,0) -- (7,0);
		\node[text width=4cm] at (9,0) {Walking link};
		\draw[->, dotted] (11,0) -- (13,0);
		\node[text width=4cm] at (15,0) {Waiting link};
		\draw[->, loosely dashdotted] (17,0) -- (19,0);
		\node[text width=4cm] at (21,0) {In-vehicle link};
		\draw[->, densely dashdotted] (23,0) -- (25,0);
		\node[text width=4cm] at (27,0) {Transfer link};
		
		\end{tikzpicture}
	\end{subfigure}
	
	\caption{An example network to illustrate the processing of Algorithm~\ref{alg:sbtsp}}
	\label{fig:trips}
	\end{figure}
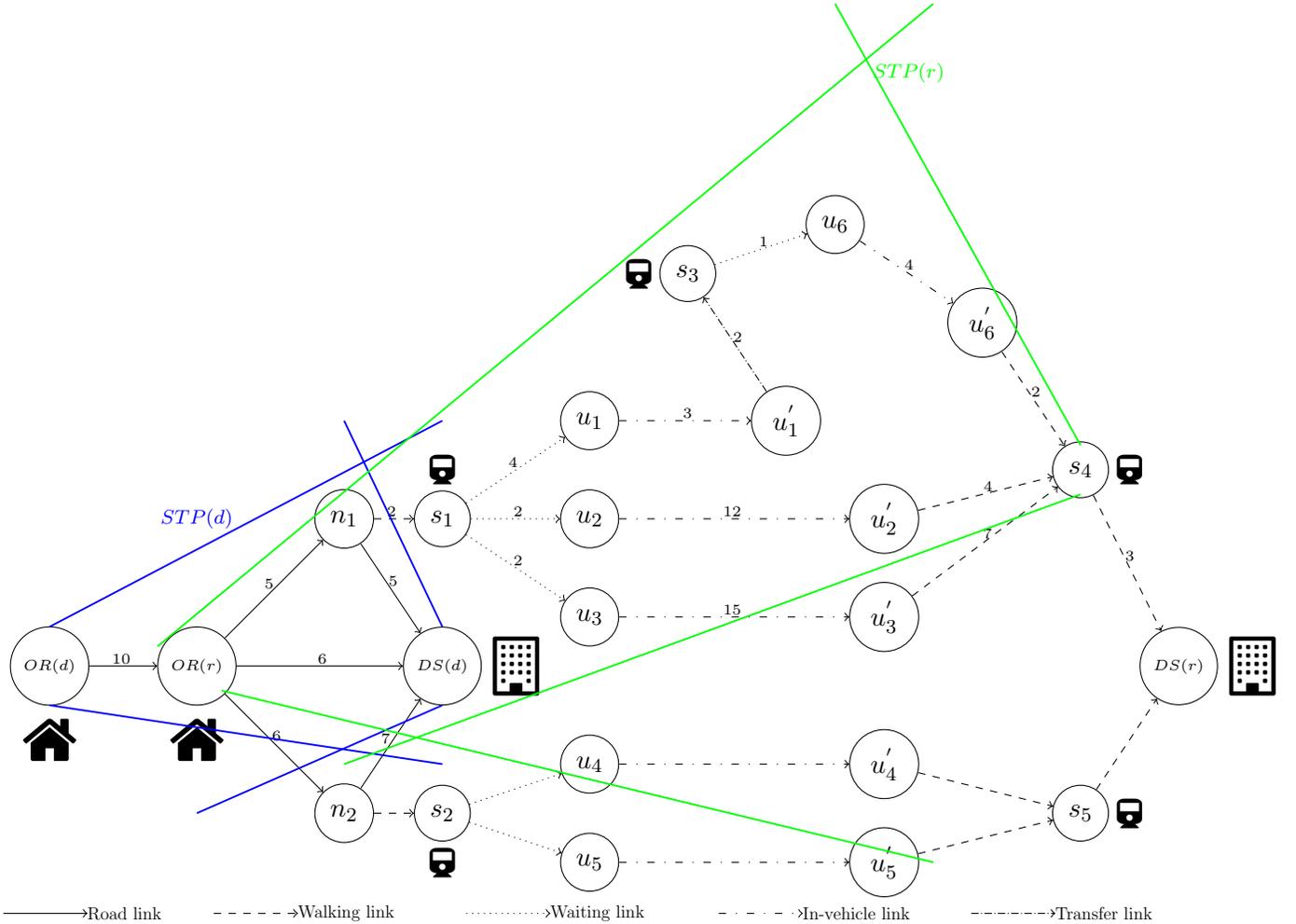

	\begin{prop}\label{prop:3}
		Algorithm~\ref{alg:sbtsp} terminates after a finite number of iterations and produces optimal transit path for riders.
	\end{prop}
	\begin{proof}
		As the number of riders $\vert R \vert $ and drivers $\vert D \vert$ are finite in the program, the number of iterations by two \emph{for} loops (Line 3 and 21) should be finite ($\vert R \vert \vert D \vert$). While finding an optimal itinerary for a rider up to current node $j$, the labels cannot be unbounded as link costs are non-negative and any node enters SEL only once. Therefore, the algorithm terminates in the finite number of iterations. Also, the algorithm satisfies (\ref{eq:bpo}), hence the labels should be optimal.
	\end{proof}
	
	\begin{rem}\label{rem:1}
		Algorithm~\ref{alg:sbtsp} will not cut off any feasible solution if and only if $\eta_a= \eta_w =  \eta_v= 1$
	\end{rem}
    Consideration of varying weights to different links may cut off feasible solutions that are not attractive to the passenger. But if we want to consider all the feasible solutions, the coefficients should be $\eta_a= \eta_w =  \eta_v= \eta_t = 1$. For practical purposes, this input should be user-specified.

\begin{prop}\label{prop:4}
	The worst case computational complexity of Algorithm~\ref{alg:sbtsp} is $\mathcal{O}(\vert R\vert (\vert D \vert $ $\vert A_T \vert\log\vert N_T \vert + \vert N_T\vert \log\vert N_T\vert))$.
\end{prop}
\begin{proof}
	Lines 2, 4-6 can be done in $\mathcal{O}(1)$ time. Assuming that the shortest path algorithm is implemented using Binary heap data structure, Line 7-20 consists of two major steps, finding the node $i$ with minimum label $\gamma^{gc}$ from SEL (Line 8), which can be done in $\mathcal{O}(\vert N_T\vert \log(\vert N_T\vert))$ and updating labels of new nodes which can be done in $\mathcal{O}( \vert A_T \vert\log(\vert N_T\vert))$ time. While updating labels of each node $j$, they are checked for possible drop-off location for rider. The  space-time prism constraints are incorporated in Line 18 and 22-25, can be done in $\mathcal{O}(1)$ time and are repeated for each driver $d \in D$. So, updating of labels should be multiplied with the cardinality of set $D$. Also, the shortest path algorithm is run for each $r \in R$, the overall complexity should be multiplied by the cardinality of set $R$. Hence, the worst case computational complexity of Algorithm~\ref{alg:sbtsp} is $\mathcal{O}(\vert R\vert (\vert D \vert \vert A_T \vert\log \vert N_T\vert + \vert N_T\vert \log\vert N_T\vert))$.
\end{proof}

\begin{rem}
	If $\vert D \vert \vert A_T \vert = \Omega(\vert N_T\vert)$, then the worst case computational complexity of Algorithm~\ref{alg:sbtsp} can be given as $\mathcal{O}(\vert R\vert \vert N_T\vert \log\vert N_T\vert))$. However, the actual cost will be much lower than the worst case. For example, a small number of nodes will be added to SEL due to the strict constraint in line 18.
\end{rem}

\subsubsection{Last mile access problem}\label{sec:lm}
In this case, a passenger can access transit for the first mile but cannot reach her destination due to the unavailability of transit service to complete the last mile. This is a common problem faced by commuters going back to home from a high-density area to a low-density area. To solve this problem, a driver also going in the same direction can pick up the passenger from alighting stop and drop her off at the respective destination. Algorithm~\ref{alg:sbtsp} can be modified by running the forward shortest path from $OR(r)$ and checking if driver $d$ driving from $OR(d)$ to a transit node $j$ can reach before bus/train arrival time at $j$, drive from that station to $DS(r)$ and then $DS(d)$ without violating the time constraints. Specifically, line 23-25 in Algorithm~\ref{alg:sbtsp} can be replaced with the following constraints:
\begin{eqnarray}
\tau^{ed}(d) + t_{OR(d)Z(j)} \le \gamma_{new}\\
\gamma_{new} + t_{Z(j)DS(r)} \le \tau^{la}(r)\\
\gamma_{new} + t_{Z(j)DS(r)} + t_{DS(r)DS(d)}  \le \tau^{la}(d)
\end{eqnarray}

\subsection{Optimal assignment of drivers to riders}\label{sec:optp}
The second step of the procedure is to find an optimal match between riders and drivers among the feasible matches found by Algorithm~\ref{alg:sbtsp}. This is formalized as an Integer Linear Program (ILP). The goal of this optimization program is to find optimal matches out of feasible matches that optimize a given objective. The optimization problem is described below:\\

\subsubsection{Variables} 
Let $\mathcal{M}_r= \{m \in \mathcal{M} : r \in m\}$ be the set of all feasible matches of $r$ present in $\mathcal{M}$. Similarly, $\mathcal{M}_d= \{m \in \mathcal{M} : d \in m\}$ be the set of all the feasible matches in $\mathcal{M}$ involving $d \in D$. The set of feasible matches can be seen as edges between riders and drivers in a \textit{bipartite graph}, in which a rider can appear more than once with different drop off location and itinerary (Figure \ref{fig:bipartite}). Let us assume a binary variable $\epsilon_{k}$, which takes the value 1, if edge $k \in \mathcal{M}$ is selected,  and 0 otherwise. Let $t_k^{drive}, t_k^{transit}, t_k^{walk}, t_k^{\#trans}, t_k^{wait}, \ \text{and}  \ t_k^{vhrs}$ be the driving time, transit time, walking time, number of transfers, waiting time, and vehicle-hrs savings associated with edge $k \in \mathcal{M}$. These parameters can be used to define different objectives for the matching optimization program. The vehicle-hrs savings $t_k^{vhrs}$ can be computed by subtracting $t_k^{drive}$ from sum of driving time when there is no ridesharing (\ref{eq:vhrs}).
\begin{equation}\label{eq:vhrs}	
t_k^{vhrs} =  t_{OR(r)DS(r)} + t_{OR(d)DS(d)} - t_k^{drive} \  \text{for} \  r, d \in k \quad \forall k \in \mathcal{M}
\end{equation}

\begin{figure}[h!]		
	\centering
	\begin{tikzpicture}[]		
	\node[shape=circle,draw=black] (1) at (0,0) {$d_1$};
	\node[shape=circle,draw=black] (2) at (5,0) {$r_1, n_1$};
	\node[shape=circle,draw=black] (3) at (5,-2) {$r_1, n_2$};
	\node[shape=circle,draw=black] (4) at (5,-4) {$r_1, n_3$};
	\node[shape=circle,draw=black] (5) at (0,-5) {$d_2$};
	\node[shape=circle,draw=black] (6) at (5,-6) {$r_2, n_1$};
	\node[shape=circle,draw=black] (7) at (5,-8) {$r_2, n_2$};
	\node[shape=circle,draw=black] (8) at (5,-10) {$r_3, n_4$};
	\path [->] (1) edge node[left, yshift=7pt, xshift =13pt] {$t^{vhrs}_{d_1, r_1, n_1}$} (2);
	\path [->] (1) edge node[left, yshift=7pt, xshift =23pt] {$t^{vhrs}_{d_1, r_1, n_2}$} (3);
	\path [->] (1) edge node[left, yshift=3pt, xshift =12pt] {$t^{vhrs}_{d_1, r_1, n_3}$} (4);
	\path [->] (5) edge node[left, yshift=0pt, xshift =16pt] {$t^{vhrs}_{d_2, r_1, n_2}$} (3);
	\path [->] (5) edge node[left, yshift=0pt, xshift =16pt] {$t^{vhrs}_{d_2, r_1, n_1}$} (2);
	\path [->] (5) edge node[left, yshift=0pt, xshift =16pt] {$t^{vhrs}_{d_2, r_2, n_1}$} (6);
	\path [->] (5) edge node[left, yshift=0pt, xshift =16pt] {$t^{vhrs}_{d_2, r_3, n_4}$} (8);
	\path [->] (5) edge node[left, yshift=7pt, xshift =13pt] {$t^{vhrs}_{d_2, r_2, n_2}$} (7);
	% Use mirror to create the braces on other side
	\draw [decorate,decoration={brace,amplitude=10pt,%mirror,
		,raise=4pt},yshift=0pt]
	(5.5,1) -- (5.5,-5) node [black,midway,xshift=0.8cm] {\footnotesize
		$r_1$};
	\draw [decorate,decoration={brace,amplitude=10pt,%mirror,
		,raise=4pt},yshift=0pt]
	(5.5,-5.5) -- (5.5,-8.5) node [black,midway,xshift=0.8cm] {\footnotesize
		$r_2$};
	\draw [decorate,decoration={brace,amplitude=10pt,%mirror,
		,raise=4pt},yshift=0pt]
	(5.5,-9.5) -- (5.5,-10.5) node [black,midway,xshift=0.8cm] {\footnotesize
		$r_3$};
	
	\end{tikzpicture}
	\caption{An illustration of bipartite graph in case of transit-based ridematching}
	\label{fig:bipartite}
\end{figure}
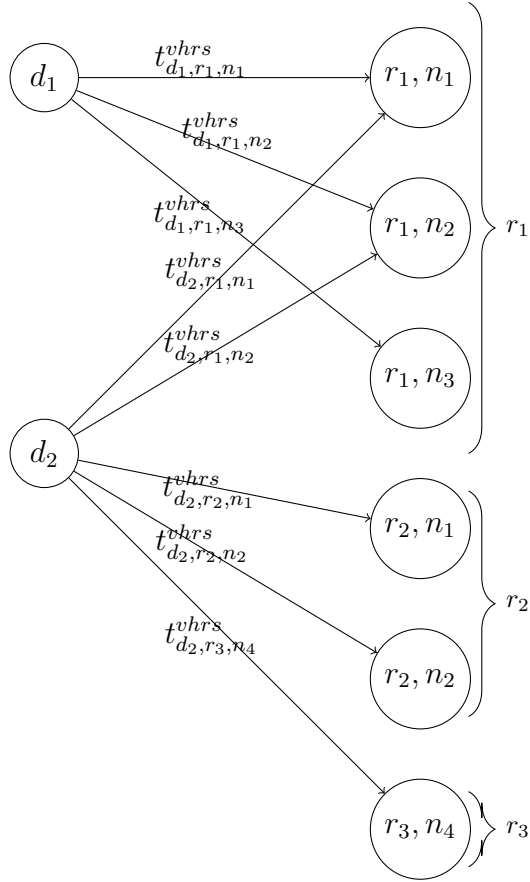

\subsubsection{Objective}  
There can be several objectives for matching riders and drivers. We consider two different objectives for this problem:

\begin{enumerate}
	\item Maximize total number of matches, $ Z_1 = \sum_{k \in \mathcal{M}}  \epsilon_k$
	\item Maximize veh-hrs savings, $ Z_2 = \sum_{k \in \mathcal{M}} t_k^{vhrs} \epsilon_k$
\end{enumerate}

\subsubsection{Constraints}  
There are two types of constraints which are given below:
\begin{enumerate}
	\item Each rider should be assigned to at most one driver.
	\begin{equation}
	\sum_{k \in \mathcal{M}_r} \epsilon_k \le 1 \quad \forall r \in R
	\end{equation}
	\item Each driver should be assigned to at most one rider.
	\begin{equation}
	\sum_{k \in \mathcal{M}_d} \epsilon_k \le 1 \quad \forall d \in D
	\end{equation}
\end{enumerate}
The overall optimization problem to optimally assign riders with drivers can ve written as follows:
\begin{equation} 
\begin{aligned}
& \underset{\epsilon}{\text{minimize}}
& & Z \\
& \text{subject to}
& & \sum_{k \in \mathcal{M}_r} \epsilon_k \le 1 \quad \forall r \in R\\
& & & \sum_{k \in \mathcal{M}_d} \epsilon_k \le 1 \quad \forall d \in D\\
& & & \epsilon_{k} = \{0, 1\} \quad \forall k \in \mathcal{M} \label{eq:10}
\end{aligned}
\end{equation}
\begin{rem}
	The solution obtained by LP relaxation of (\ref{eq:10}) produces an integral solution.
\end{rem}
\begin{proof} (Informal)
    Due to the repetition of riders in different nodes, the bipartite graph in this case (Figure \ref{fig:bipartite}) is different than the graph considered in a classic bipartite matching. However, the solution obtained by LP relaxation of (\ref{eq:10}) still produces an integral solution. This is because the matching polytope still satisfies the totally unimodularity condition. As each edge is unique, each column of the matrix produced by the set of constraints will still have at most two non-zero values. Of course, each entry of the matrix would be either 0 or 1, and rows of the matrix can still be partitioned into two sets with non-zero values in separate sets.
\end{proof}

	\subsection{A rolling horizon algorithm for dynamic transit-based ridesharing}\label{sec:rh}
The ridesharing matching problem is dynamic in nature as requests arrive continuously overtime throughout the day. This requires solving the given problem several times during the day since future requests are unknown and drivers and riders enter and leave the system continuously over time. This type of uncertainty is usually managed using \textit{rolling horizon strategy} (\citet{Agatz2011, Wang2017}). Using this strategy, the matching problem is solved in several instances, and matches are not finalized until necessitated by the deadline. The frequency of optimization can be determined using two approaches-\textit{event driven} or \textit{fixed time step} approach (\citet{Najmi2017}). The event-driven approach initiates optimization at each time a new request is received. However, this may lead to synchronization issues if a new request is received before solving the previous problem.  For the framework proposed in this study, it is efficient to use a fixed time-step approach with predefined time steps. Pseudocode for using a rolling horizon approach to solve transit-based ridesharing matching problem is proposed in Algorithm~\ref{alg:rh}. The algorithm starts with initializing set $D^{sys}(\tau)$ and $R^{sys}(\tau)$ representing driver and rider requests currently in the system. At each time step $\kappa$, we consider all new rider and driver requests ($D^{new}$ and $R^{new}$ respectively) received by the platform during the time interval of length $\sigma$. We add these new requests to the set of existing requests $D^{sys}(\tau)$ and $R^{sys}(\tau)$. We then maintain a set of feasible match $\mathcal{M}$ for $d \in D^{sys}(\tau)$ and $r \in R^{sys}(\tau)$. To make it computationally efficient, we run Algorithm~\ref{alg:sbtsp} for new riders and drivers only (line 12-13). The new feasible matches found using Algorithm~\ref{alg:sbtsp} are added to the set $\mathcal{M}$. Then, an optimization problem is called to decide a set of optimal matches $\mathcal{M^{*}}$ for given feasible matches $\mathcal{M}$. After this, a subset of matches  $\mathcal{M}^{\text{fin}} \subseteq \mathcal{M^{*}}$ are finalized for the requests which are expiring in next time step. A lead time $\mu$ is also subtracted from $\tau^{ld}(p)$ to allow some flexibility to riders and drivers to get ready for the trip.  The requests which are finalized and the ones which did not found any optimal matches and expiring in the next step are also removed from the system. Then, the horizon is rolled over to the next time step and the process is repeated. The results of a simulation experiment are given in  \cref{sec:rhe}

\begin{algorithm}[H]
	\caption{Pseudocode for dynamic ridesharing using rolling horizon policy}
	\label{alg:rh}
	\begin{algorithmic}[1]
		\Procedure{RollingHorizon}{}
		\State $D^{sys}(\tau) \gets \{\}$ \Comment{set of drivers in the system at time $\tau$} 
		\State $R^{sys}(\tau) \gets \{\}$ \Comment{set of riders in the system at time $\tau$} 
		\State $\mathcal{M}(\tau) \gets \{\}$ \Comment{set of feasible matches at time $\tau$} 
		\State $\mathcal{M}^{\text{fin}}(\tau) \gets \{\}$ \Comment{set of final optimal matches at time $\tau$} 
		\State $\tau \gets 0$
		\For{$\kappa \in \{1, 2, 3, ..., n\}$} \Comment{time steps} 
		\State $R^{new} \gets \{r \in R| \tau \le \tau^{at}(r) \le \tau + \kappa\sigma\}$ \Comment{set of new riders in time step $k$} 
		\State $D^{new} \gets \{d \in D| \tau \le \tau^{at}(d) \le \tau + \kappa\sigma\}$ \Comment{set of new drivers in time step $k$} 
		\State $D^{sys} \gets D^{sys} \cup D^{new}$
		\State $R^{sys} \gets R^{sys} \cup R^{new}$
		\State $\mathcal{M} \gets \mathcal{M} \cup \Call{RiderDriverPotentialMatch}{R^{new}, D^{sys}, \mathcal{G_R}, \mathcal{G_T}, \eta}$
		\State $\mathcal{M} \gets \mathcal{M} \cup \Call{RiderDriverPotentialMatch}{R^{sys}, D^{new}, \mathcal{G_R}, \mathcal{G_T}, \eta}$
		\State $\mathcal{M}^{*} \gets \Call{\textsc{OPTIMIZE}}{\mathcal{M}}$ \Comment{set of current optimal matches} 
		\For{$(r, d, j, \xi) \in \mathcal{M}^{*}$}
		\If{$\tau^{ld}(r) - (\tau + \kappa\sigma) \le \mu$ \Or $\tau^{ld}(d) - (\tau + \kappa\sigma) \le \mu$}
		\State $\mathcal{M}^{\text{fin}} \gets \mathcal{M}^{\text{fin}} \cup \{(r, d, j, \xi)\}$
		\State $\mathcal{M} \gets \mathcal{M} \backslash \{m \in \mathcal{M} \ | \ (r, d, j, \xi) \in \mathcal{M}^{*}  \ \& \  r \in m \ \&  \ d \in m \}$
		\State $D^{sys} \gets D^{sys}\backslash \{d\}, R^{sys} \gets R^{sys}\backslash \{r\}$
		\EndIf
		\EndFor
		\State $\tau \gets \tau  + \kappa\sigma$
		\EndFor		
		\EndProcedure        
	\end{algorithmic}
\end{algorithm}

		\section{Numerical experiments}\label{sec:exp}
    In this section, we present results from simulation experiments to get more insights and analyze the benefits of the proposed transit-based ridesharing program. We implemented Algorithm~\ref{alg:sbtsp} and \ref{alg:rh} in Python 2 on a standard PC with Core i5 3.00 GHz processor with 16 GB RAM.        
	\subsection{Simulation setup}
    The simulation environment is based on the 2010 activity-based travel demand model for Twin Cities, MN developed by Metropolitan Council (\citet{metcouncilabm}). The Twin Cities metropolitan region is spread over 8,120 mi$^2$ area with a population of about 3.55 million. There are more than 11 million trips occurring every day. The activity-based model takes input from travel behavior surveys and outputs simulated trips for the entire metropolitan region. The model output used for this study consists of trip origin, destination, activity purpose, the mode used, departure time, and arrival time. For this ridesharing program, we consider only the commuting and school trips during morning peak hours (6:00 AM to 9:00 AM). The departure and arrival time of each trip in this data is available on a 30-minute scale, so we added a uniform random number between 0 to 30 minutes to represent the variation in the departure time of trips. The Twin Cities transportation network has 3030 traffic analysis zones (TAZs), 23,812 nodes, and 56,688 links. Using a static traffic assignment program (\citet{pramesh_kumar_2017_1045573}), we calculated the value of equilibrium travel time on links for an hourly demand obtained from the above model. For each trip, the value of the shortest path travel time is increased by a certain proportion, known as \textit{time flexibility}. The latest arrival time $\{\tau^{la}(p), p \in P\}$ at passengers' destination are obtained by adding this increased travel time to the earliest departure time $\{\tau^{ed}(p), p \in P\}$. Then a uniform random number between 0 and 60 minutes is further subtracted from $\{\tau^{ed}(p), p \in P\}$ to calculate the value of announcement time $\{\tau^{at}(p), p \in P\}$ of trips. After this, unlike other studies, we explicitly selected trips with the transit first-mile access problem (i.e., no transit stop within a buffer distance of 0.75 mi). This is important to assess the efficiency of a matching algorithm designed to solve the first-mile problem. The driver or rider role to a given trip is also assigned randomly based on the driver-rider ratio. Other specifications are given in Table~\ref{tab:2}.\\
	
	\begin{table}[h!]
		\centering
		\caption{Characteristics of simulation setup}
		\label{tab:2}
		\begin{tabular}{ll}
			\hline        
			Trips & Commuting and school trips\\
			& having first mile access problem\\            
			Driver-rider ratio &    1\\
			Acceptable walking distance to/from bus stop &    0.75 mi\\
			Acceptable transfer distance for walking  &    0.25 mi\\
			
			Walking speed of riders &    3 mi/hr\\            
			Maximum wait time at transit stop &    10 min\\
			$t^{ser}$ &    2 min\\
			$\eta$ &    \{1.5, 2.0, 1.0, 2.0\}\\
			Vehicle capacity &    2 seats\\
			\hline    
		\end{tabular}
	\end{table}

        Metro Transit is the primary transit agency in the Twin Cities region, offering an integrated network of buses, light rails, and a commuter train. The schedule-based transit network is created using the GTFS data obtained from \cite{gtfsmn} with the help of an R script (\cite{pramesh_kumar_2019_2531085}). The script produces different types of transit network links ${A_T}$ based on the set of pre-defined parameters. The time to traverse an access link $(i, j)$ is calculated by dividing the distance between $i$ and $j$ and speed of walking, which we considered as 3 mi/hr. The Metro Transit GTFS feed contains 191 routes, 13,672 stops, and 9,042 vehicle trips. The number of different types of transit network links generated is given in Table~\ref{tab:3}. The standard implementation of schedule-based transit shortest path using binary heaps is also shared as open-source code (\cite{alireza_khani_2018_1341312}).\\

	\begin{table}[h!]
		\centering
		\caption{Twin Cities schedule-based transit network}
		\label{tab:3}
		\begin{tabular}{ll}
			\hline        
			Number of nodes    &    490,112\\
			Number of in-vehicle links &        481,065\\
			Number of waiting transfer links         &    241,187\\
			Number of walking transfer links  &        5,082,747\\
			Number of access/egress links &        8,496,605 \\
			Total number of links          &    14,301,604 \\
			Computational time to generate transit network    &    8 min\\            
			\hline    
		\end{tabular}
	\end{table}

    In the following subsections, we first present the results for a static matching problem, in which case the requests are known before the start of the day. Then, a dynamic simulation experiment results using the rolling horizon policy is presented in \cref{sec:rhe}. The sensitivity analysis of different parameters, such as participation rate, rider, and driver time flexibility, and the driver-rider ratio is also presented. The results are computed for different objective functions given in \cref{sec:optp}.
	
	\subsection{Static simulation}
    For this experiment, the driver-rider ratio is equal to 1.0, i.e., an equal number of drivers and riders in the system and a participation rate of 1\% (1,087 riders and 1,092 drivers) among travelers facing the transit first-mile access problem. We provide time flexibility of 40\% to the drivers and 80\% to the riders as transit travel time is likely to be higher than the auto travel time. To avoid any bias, we draw five different samples and calculate the average of the results obtained. The results obtained using different samples with different matching objectives are given in Table~\ref{tab:4}. We can observe that for 1,088 drivers and 1,097 riders, Algorithm~\ref{alg:sbtsp} took an average of 7.7 minutes and produced 12,338 feasible matches. It only took a fraction of a second to solve the matching optimization problem (\ref{eq:10}). The average number of optimal matches found using objective $Z_1$ exceeds objective $Z_2$ by a small number. However, the average veh-hrs savings found using objective $Z_2$ is significantly more than the value found using $Z_1$. As $Z_2$ is about maximizing veh-hrs savings, the average detour time of drivers found using objective $Z_1$ is more than $Z_2$ by 2 minutes. On the other hand, the average shared time between a rider and a driver up to a transit station and average transit time using objective $Z_2$ is more than $Z_1$. The average walking time of riders using both objectives is almost similar to without a transfer, the walking component of the transit trip should be between the destination and closest transit station only. The waiting time is negligible which shows only a few optimal trips have a waiting transfer, which is assigned a higher weight in the algorithm. \\

	\begin{table}[]
		\caption{Static simulation results}
		\label{tab:4}
		\resizebox{\textwidth}{!}{%
			\begin{tabular}{|p{5cm}|l|l|l|l|l|l|l|l|l|l|l|l|}		
				\hline
				\textbf{Sample number} & \multicolumn{2}{l|}{\textbf{1}} & \multicolumn{2}{l|}{\textbf{2}} & \multicolumn{2}{l|}{\textbf{3}} & \multicolumn{2}{l|}{\textbf{4}} & \multicolumn{2}{l|}{\textbf{5}} & \multicolumn{2}{l|}{\textbf{Average}} \\ \hline
				\textbf{Number of participants} & \multicolumn{2}{l|}{1,120 D, 1,059 R} & \multicolumn{2}{l|}{1,035 D, 1,144 R} & \multicolumn{2}{l|}{1,096 D, 1,083 R} & \multicolumn{2}{l|}{1,092 D, 1,087 R} & \multicolumn{2}{l|}{1,099 D, 1,080 R} & \multicolumn{2}{l|}{1,088 D, 1,097 R} \\ \hline
				\textbf{Run time of Algorithm \ref{alg:sbtsp} (min)} & \multicolumn{2}{l|}{7.55} & \multicolumn{2}{l|}{7.46} & \multicolumn{2}{l|}{7.64} & \multicolumn{2}{l|}{7.47} & \multicolumn{2}{l|}{8.4} & \multicolumn{2}{l|}{7.704} \\ \hline
				\textbf{Number of feasible matches found} & \multicolumn{2}{l|}{22,397} & \multicolumn{2}{l|}{10,415} & \multicolumn{2}{l|}{9,936} & \multicolumn{2}{l|}{10,021} & \multicolumn{2}{l|}{8,924} & \multicolumn{2}{l|}{12,338} \\ \hline
				\textbf{Number of riders and drivers in feasible matches} & \multicolumn{2}{l|}{364 R, 604 D} & \multicolumn{2}{l|}{218 R, 344 D} & \multicolumn{2}{l|}{178 R, 299 D} & \multicolumn{2}{l|}{182 R, 334 D} & \multicolumn{2}{l|}{187 R, 346 D} & \multicolumn{2}{l|}{263 R, 386 D} \\ \hline
				\textbf{Objective} & $Z_1$ & $Z_2$ & $Z_1$ & $Z_2$ & $Z_1$ & $Z_2$ & $Z_1$ & $Z_2$ & $Z_1$ & $Z_2$ & $Z_1$ & $Z_2$ \\ \hline
				\textbf{Run time of optimization program} (sec) & 0.1 & 0.12 & 0.08 & 0.08 & 0.07 & 0.08 & 0.07 & 0.15 & 0.07 & 0.09 & 0.09 & 0.1 \\ \hline
				\textbf{Optimal number of matches found} & 319 & 316 & 192 & 192 & 155 & 154 & 166 & 166 & 168 & 166 & 200 & 198 \\ \hline
				\textbf{Total vehicle-hrs savings} & 107.4 & 122.47 & 55.97 & 67.03 & 48.81 & 57.15 & 48.14 & 55.97 & 47.17 & 58.7 & 61.5 & 71.6 \\ \hline
				\textbf{Average driver detour time (min)} & 10.12 & 7.65 & 11.58 & 8.82 & 11.39 & 8.48 & 11.78 & 9.08 & 11.95 & 8.29 & 11.36 & 9.2 \\ \hline
				\textbf{Shared time between driver and rider (min)} & 29.54 & 29.9 & 27.86 & 28.1 & 29.54 & 29.44 & 28.36 & 28.24 & 28.12 & 28.52 & 28.68 & 34.4 \\ \hline
				\textbf{Average transit time (min)} & 6.06 & 6.41 & 6.44 & 6.59 & 6.59 & 7.64 & 5.43 & 5.76 & 5.29 & 6.32 & 5.29 & 6.32 \\ \hline
				\textbf{Average walking time (min)} & 10.75 & 10.85 & 10.1 & 11 & 11 & 10.85 & 10.24 & 10.01 & 10.57 & 10.46 & 10.532 & 10.452 \\ \hline
				\textbf{Average waiting time (min)}& 0.5 & 0.63 & 0.59 & 0.61 & 0.61 & 0.64 & 0.58 & 0.54 & 0.463 & 0.4 & 0.5486 & 0.574 \\ \hline
			\end{tabular}
		}
	\end{table}

	\FloatBarrier
    To closely examine different components of individual driver and rider-trips, we created a stacked bar chart for the average value of different travel time components (Figure \ref{fig:tt}). Figure \ref{fig:tt}a shows the driver trips, and Figure \ref{fig:tt}b shows the rider-trips. A particular ID in both panels represents a match using objective $Z_2$. This match ID can be compared in both plots to see the relationship between different components of the trip time. As we can see drivers spent most of the time driving a rider to a station. Drivers also spent a significant amount of time driving from their origin to a rider's origin. The driving time from station drop off location to drivers' destination is low as compared to the other two components of the drivers' itinerary. The average driver detour time of  9.2 minutes (Table~\ref{tab:4}) shows that while maximizing the vehicle-hours savings, drivers are assigned to riders whose drop-off transit stations are close to drivers' destinations. \\
    
       Looking at rider itineraries, for most of the riders, the drive time from the origin location to the transit station is a significant component of the trip followed by in-vehicle time and then walking time. Walking time is higher in most of the itineraries because of the unavailability of the actual destination location in the data due to which a walking component from the centroid of TAZ to alighting stop is specified. We can observe that most of the riders are assigned a trip without any walking and waiting transfer to the other routes due to their higher weight in the algorithm. About 31\% of the riders were found to have one or two transfers on their trip.

	\begin{figure}[H]
		\hfill
		\begin{subfigure}[t]{\textwidth}            
			\includegraphics[width=0.66\textheight]{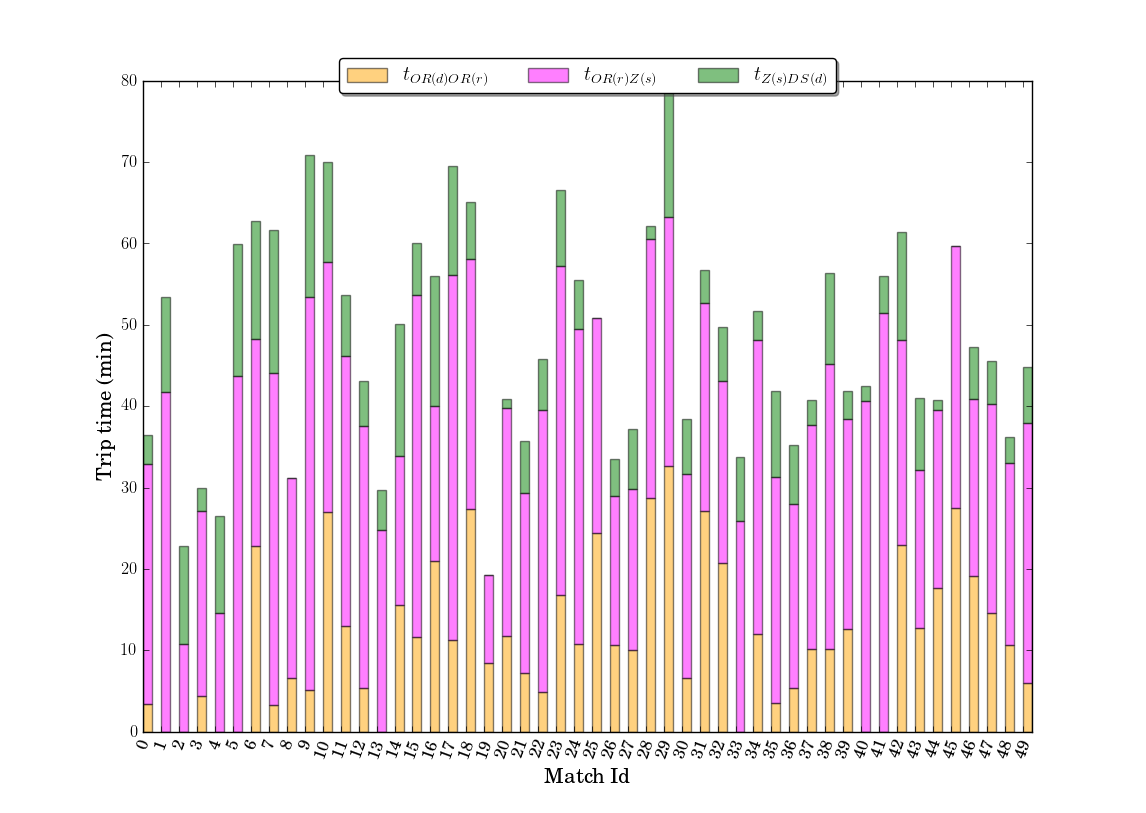}        
			\caption{Driver trip components}
		\end{subfigure}
		\begin{subfigure}[t]{\textwidth}            
			\includegraphics[width=0.66\textheight]{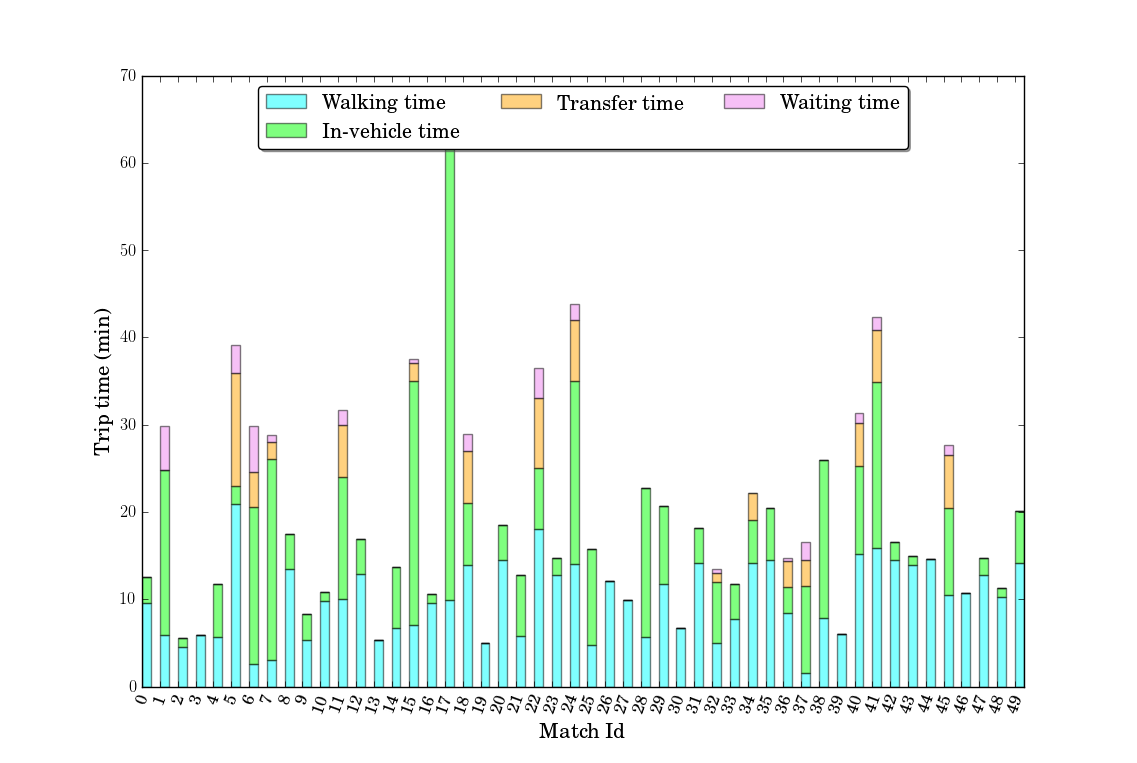}    
			\caption{Riders trip components}
		\end{subfigure}
		\caption{Different components of travel time for riders and drivers}
		\label{fig:tt}
	\end{figure}

	\subsection{Sensitivity to time flexibility}
    To better understand the effect of time flexibility, we performed a sensitivity analysis on the time flexibility of the drivers and riders. We kept the participation rate 1\% and a driver-rider ratio equal to 1.0, similar to previous tests. For varying rider and driver time flexibility (25 \% to 150\%), the results are compiled for computational time, number of feasible matches, number of optimal matches, and vehicle-hrs savings using both objectives. Figure~\ref{fig:run}a shows the effect of varying time flexibility on the computational time of Algorithm~\ref{alg:sbtsp}. With more time flexibility in rider and driver travel time, the runtime of Algorithm~\ref{alg:sbtsp} increases. The increase in computational time with an increase in the drivers' flexibility is due to an increase in the number of feasible rider matches. The increase in rider time flexibility can significantly increase the computational time due to a higher number of nodes in space-time prism of riders resulting in a higher number of feasible matches. The effect of rider time flexibility on computational time is more than driver time flexibility. Figure~\ref{fig:run}b shows the number of feasible matches found for an increase in rider and driver time flexibility. In this figure, we can observe an abrupt increase in the number of feasible matches with the increase in rider and driver time flexibility. This is due to the increase in the number of feasible transit stops and transit trip itineraries.

	\begin{figure}[h!]
		\begin{subfigure}[t]{0.50\textwidth}            
			\includegraphics[width=0.42\textheight]{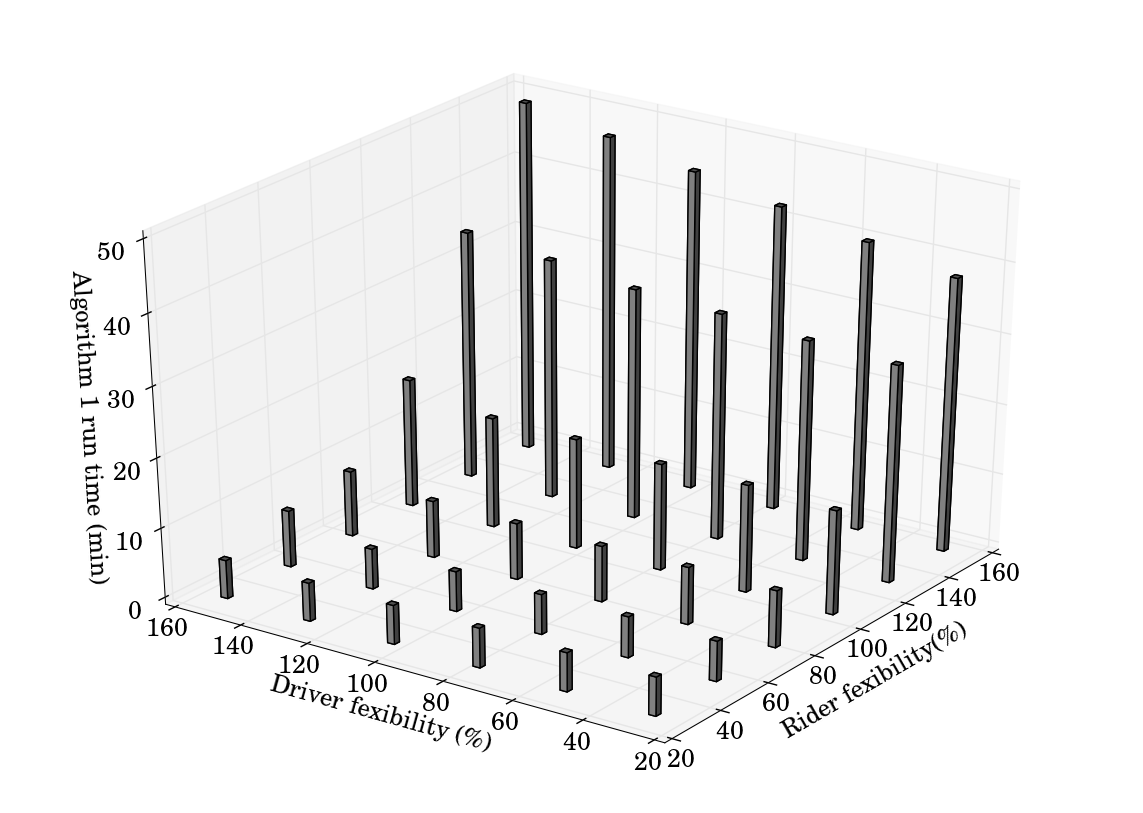}        
			\caption{Run time of Algorithm~\ref{alg:sbtsp}}
		\end{subfigure}
		\hfill
		\begin{subfigure}[t]{0.50\textwidth}            
			\includegraphics[width=0.42\textheight]{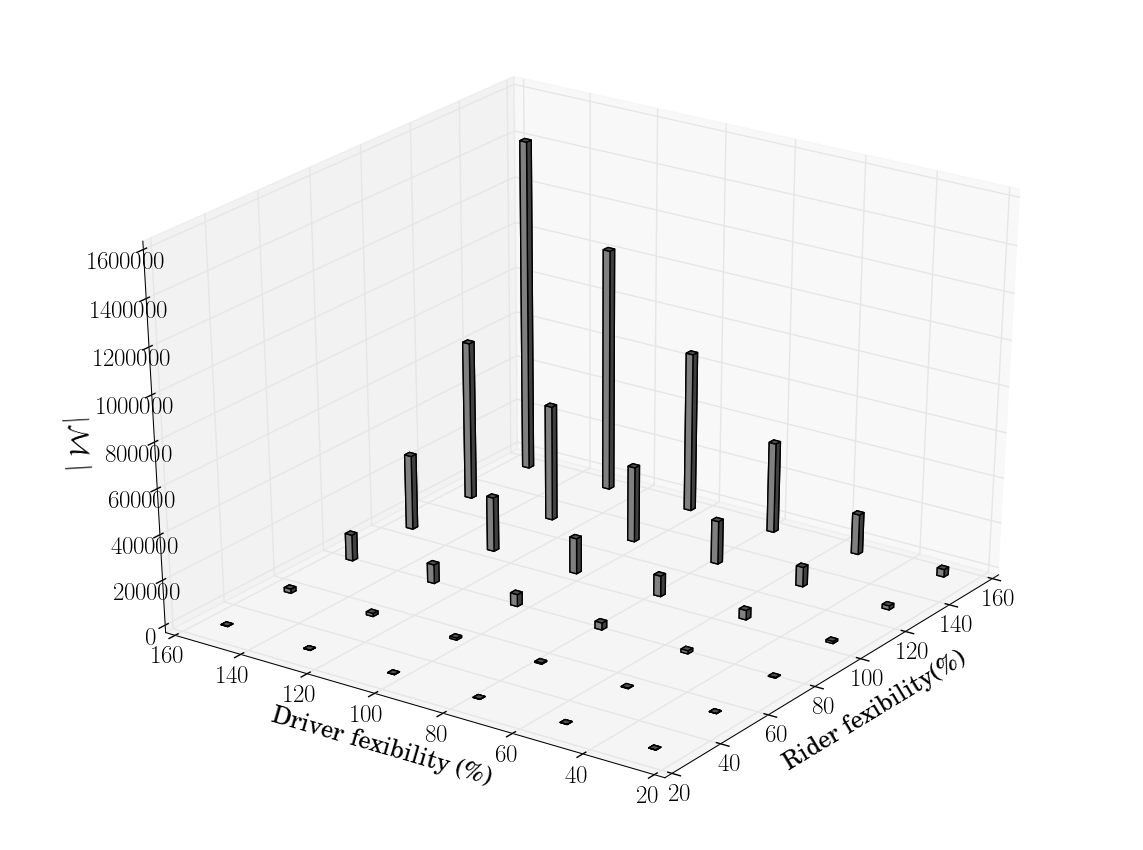}    
			\caption{Number of feasible matches}
		\end{subfigure}
		\caption{The effect of time flexibility of riders and drivers on a) run time of Algorithm~\ref{alg:sbtsp} and b) number of feasible matches found}
		\label{fig:run}
	\end{figure}
	
    Figure~\ref{fig:opt} and Figure~\ref{fig:vhs} show the effect of varying rider and driver time flexibility on the number of optimal matches found and vehicle-hrs savings using different objectives. In general, the number of optimal matches found increase with an increase in the time flexibility of both riders and drivers. As expected, the number of optimal matches are more in case of objective $Z_1$ in comparison to $Z_2$. Similarly, the veh-hrs savings in case of objective $Z_2$ is higher in comparison to $Z_1$. Interestingly, the effect of rider flexibility on the number of optimal matches found is more in comparison to driver time flexibility. However, in the case of veh-hrs savings, the driver time flexibility has more effect on objective $Z_1$, and rider time flexibility has more effect on objective $Z_2$.

	\begin{figure}[h!]
		
		\begin{subfigure}[t]{0.50\textwidth}            
			\includegraphics[width=0.36\textheight]{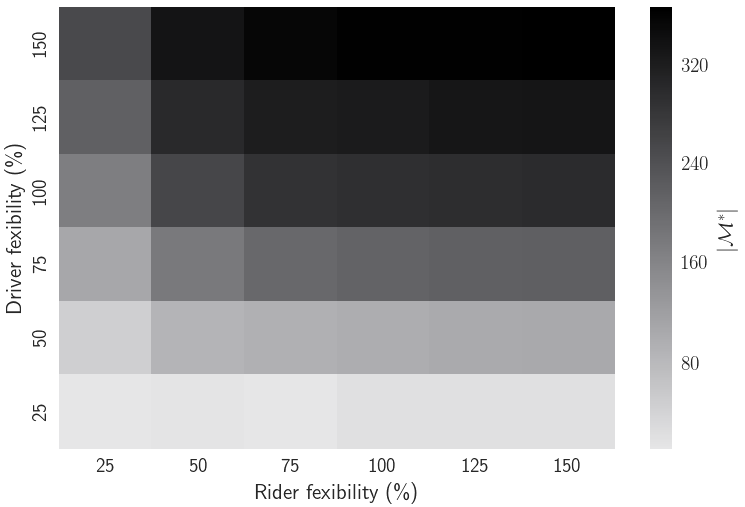}        
			\caption{Objective $Z_1$}
		\end{subfigure}
		\hfill
		\begin{subfigure}[t]{0.50\textwidth}            
			\includegraphics[width=0.36\textheight]{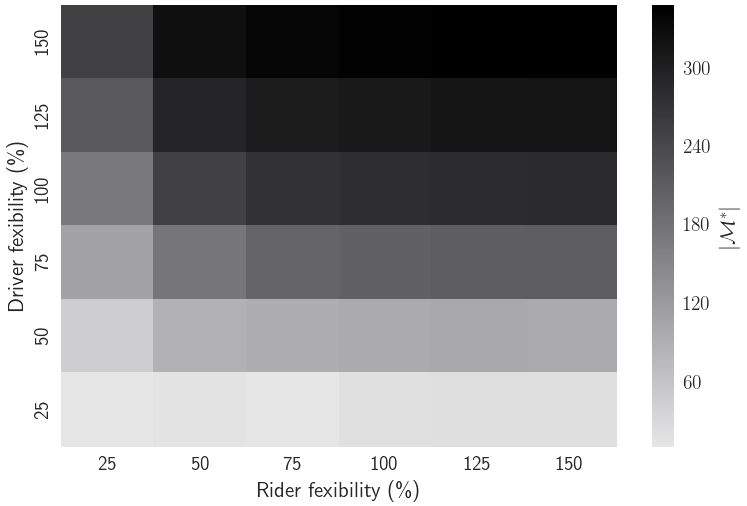}        
			\caption{Objective $Z_2$}
		\end{subfigure}    
		\caption{The effect of time flexibility of riders and drivers on optimal number of matches}
		\label{fig:opt}
	\end{figure}
	
	\begin{figure}[h!]
		\begin{subfigure}[t]{0.50\textwidth}            
			\includegraphics[width=0.36\textheight]{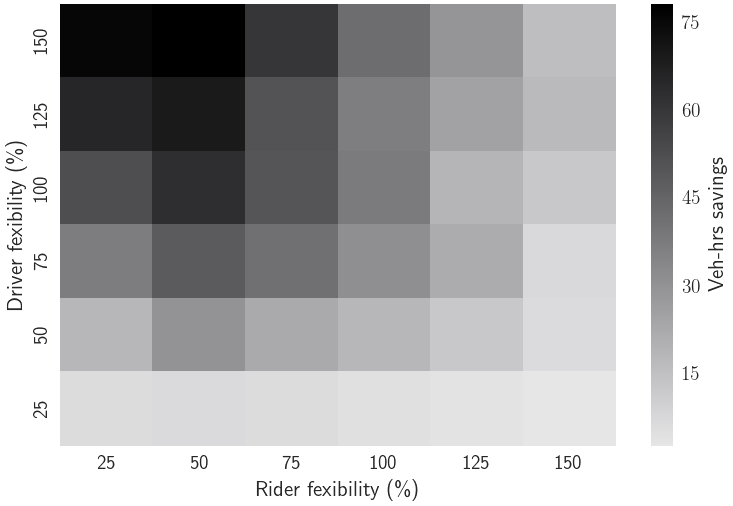}    
			\caption{Objective $Z_1$}
		\end{subfigure}        
		\hfill
		\begin{subfigure}[t]{0.50\textwidth}            
			\includegraphics[width=0.36\textheight]{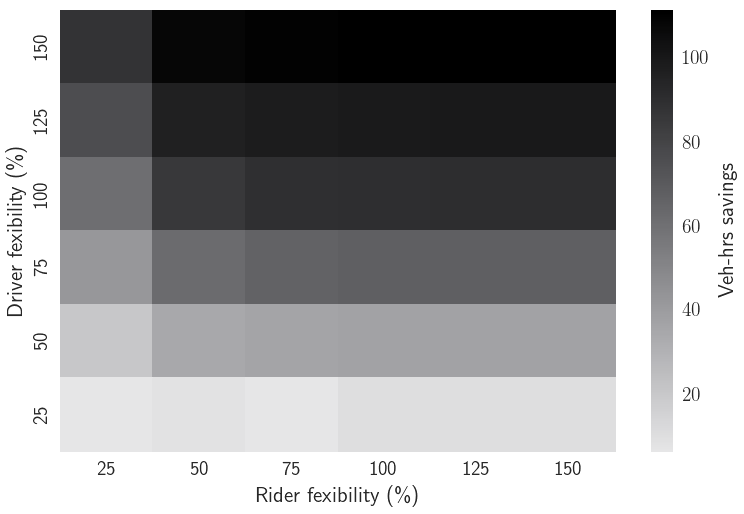}        
			\caption{Objective $Z_2$}
		\end{subfigure}
		\caption{The effect of time flexibility of riders and drivers on veh-hrs savings}
		\label{fig:vhs}
	\end{figure}

	\subsection{Sensitivity to participation rate}
    Using 150\% time flexibility for riders, 50\% time flexibility for drivers, and the driver-rider ratio of 1.0, results for different participation rates varying from 0.5\% to 2\% are stated in Table~\ref{tab:5}. With increasing participation rate, run time of Algorithm~\ref{alg:sbtsp}, the number of feasible matches, the percentage of drivers and riders matched, and veh-hrs savings increase. With a 2\% participation rate, the amount of veh-hrs savings can be as high as 230.8 veh-hrs using objective $Z_2$. This experiment shows that a higher participation rate can result in batter matching results. 
	\begin{table}[]
		\caption{Sensitivity to participation ratio}
		\label{tab:5}
		\resizebox{\textwidth}{!}{%
			\begin{tabular}{|l|l|l|l|l|l|l|l|l|}
				\hline
				Participation \% &  $\vert \mathcal{M}\vert $ & \multicolumn{2}{l|}{\% drivers matched} & \multicolumn{2}{l|}{\% riders matched} & \multicolumn{2}{l|}{veh-hrs savings} \\ \hline
				&    & $Z_1$ & $Z_2$ & $Z_1$ & $Z_2$ & $Z_1$ & $Z_2$ \\ \hline
				0.5  & 32,377 & 22.44 & 21.91 & 24.28 & 23.71 & 25.53 & 36.85 \\ \hline
				1 &  131,405 & 28.09 & 27.45 & 28.64 & 27.99 & 65.68 & 93.81 \\ \hline
				2 &  676,671 & 30.68 & 30.21 & 30.09 & 28.55 & 161.06 & 230.48 \\ \hline
			\end{tabular}
		}
	\end{table}
	\subsection{Sensitivity to driver-rider ratio}
	To see how the number of riders and drivers in the system affect results, we performed a sensitivity analysis on the driver-rider ratio. The time flexibility was kept 80\% for riders and 50\% for drivers and the participation rate was kept 1\%. The driver-rider role in the data was assigned randomly varying from 0.5 to 4.0. As we can see in Figure~\ref{fig:run2} that with more number of riders in the system, the runtime of Algorithm~\ref{alg:sbtsp} and number of feasible matches, $\vert \mathcal{M} \vert $ is higher. The highest value of both quantities was observed when the number of drivers and riders are equal in the system.\\
	
	\begin{figure}[h!]
		\begin{subfigure}[t]{0.50\textwidth}            
			\includegraphics[width=0.35\textheight]{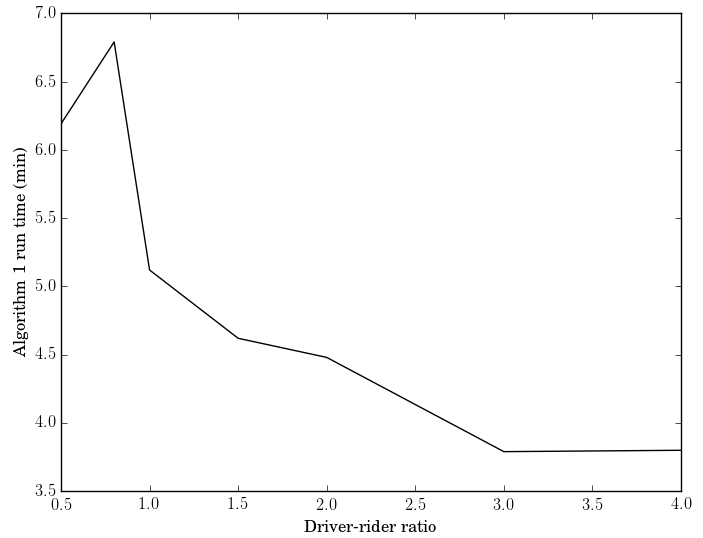}        
			\caption{Run time of Algorithm~\ref{alg:sbtsp}}
		\end{subfigure}
		\hfill
		\begin{subfigure}[t]{0.50\textwidth}            
			\includegraphics[width=0.36\textheight]{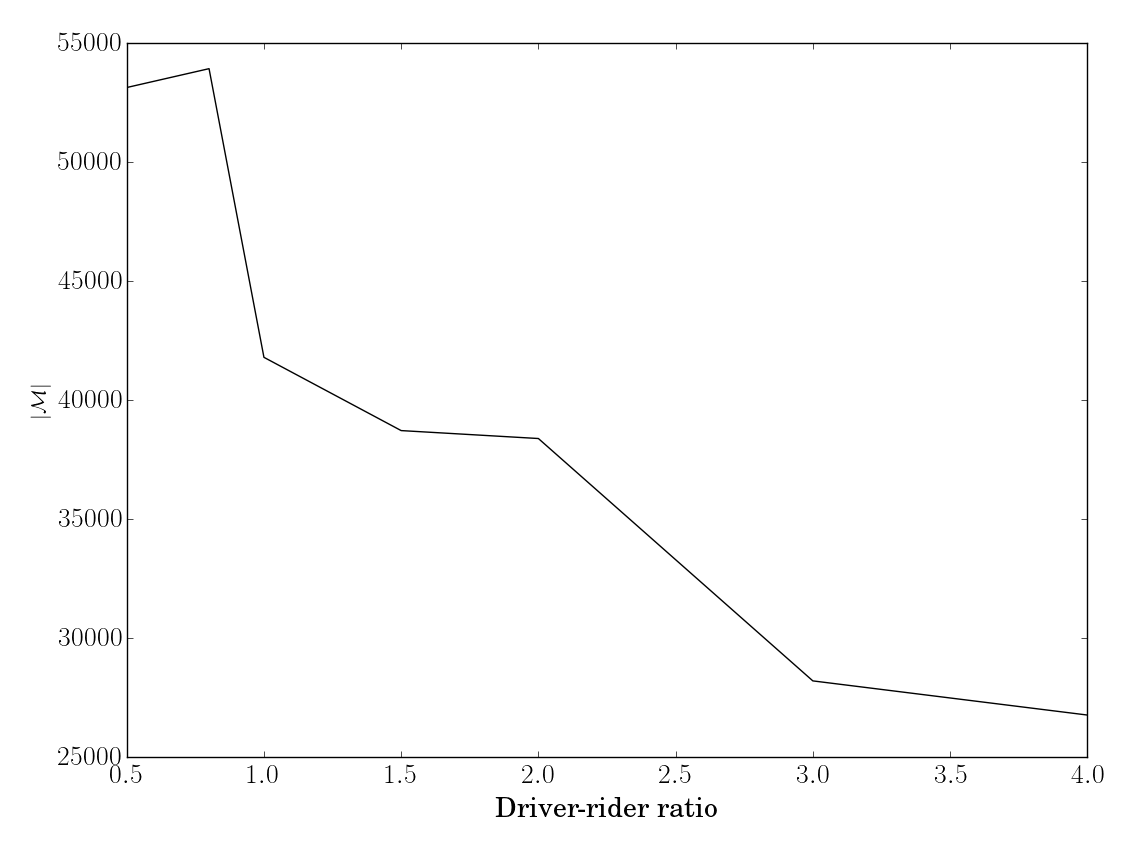}    
			\caption{Number of feasible matches}
		\end{subfigure}
		\caption{The effect of driver-rider ratio on a) run time of Algorithm~\ref{alg:sbtsp} and b) number of feasible matches found}
		\label{fig:run2}
	\end{figure}

	\begin{figure}[h!]
		
		\begin{subfigure}[t]{0.50\textwidth}            
			\includegraphics[width=0.36\textheight]{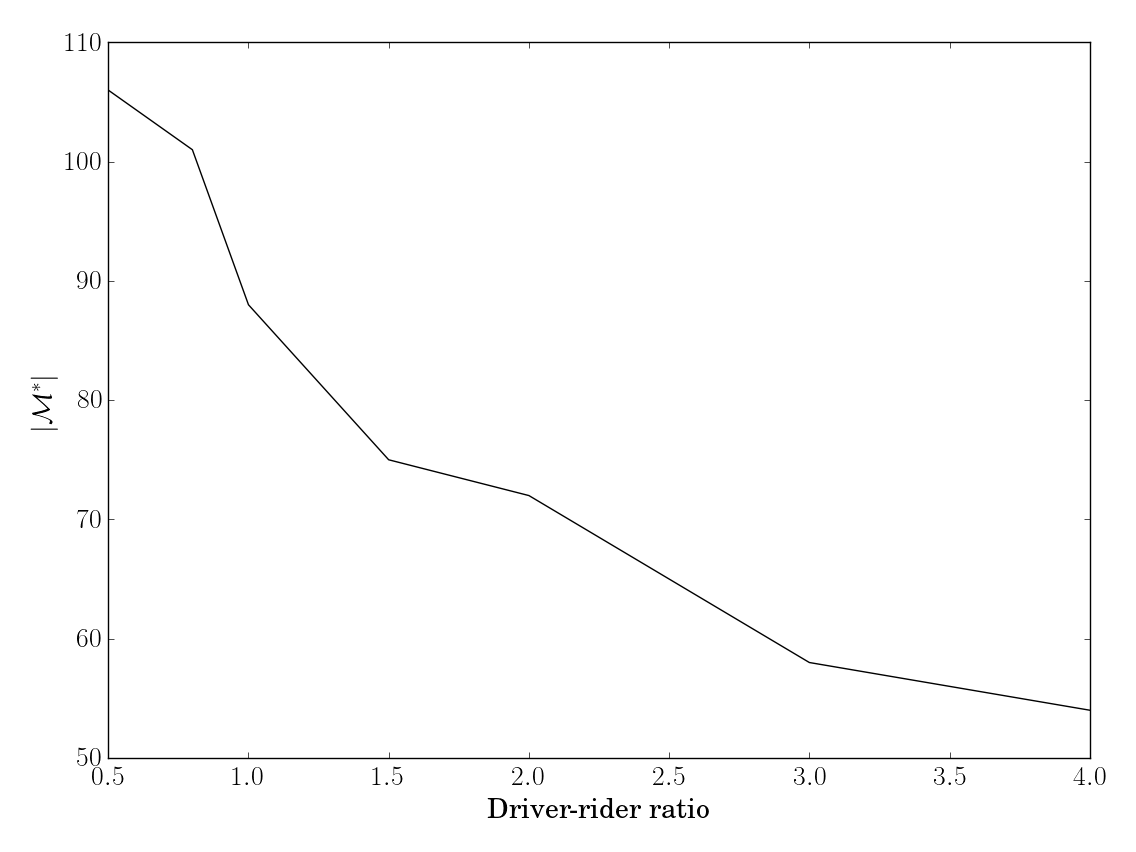}        
			\caption{Objective $Z_1$}
		\end{subfigure}
		\hfill
		\begin{subfigure}[t]{0.50\textwidth}            
			\includegraphics[width=0.34\textheight]{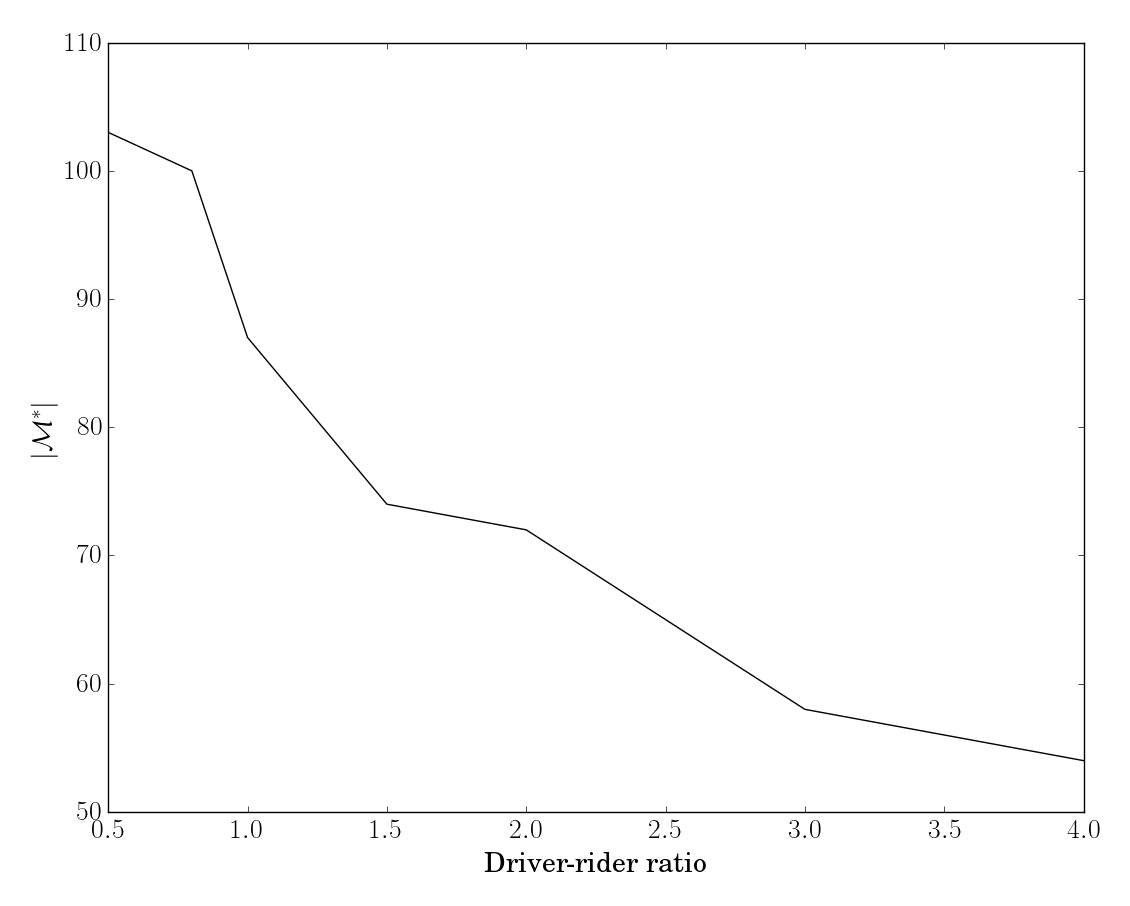}        
			\caption{Objective $Z_2$}
		\end{subfigure}    
		\caption{The effect of driver-rider ratio on optimal number of matches found}
		\label{fig:opt2}
	\end{figure}
	
    A similar trend can be observed in case of the optimal number of matches and veh-hrs savings (Figure~\ref{fig:opt2} and Figure~\ref{fig:vhs2}). With a higher number of riders in the system, more riders can be matched resulting in higher matching rate and higher veh-hrs savings. Figure~\ref{fig:opt2} shows both objectives result in an almost similar matching rate. However, objective $Z_2$ results in significantly higher veh-hrs savings in comparison to objective $Z_1$(Figure~\ref{fig:vhs2}).\\

	\begin{figure}[h!]
		\begin{subfigure}[t]{0.50\textwidth}            
			\includegraphics[width=0.37\textheight]{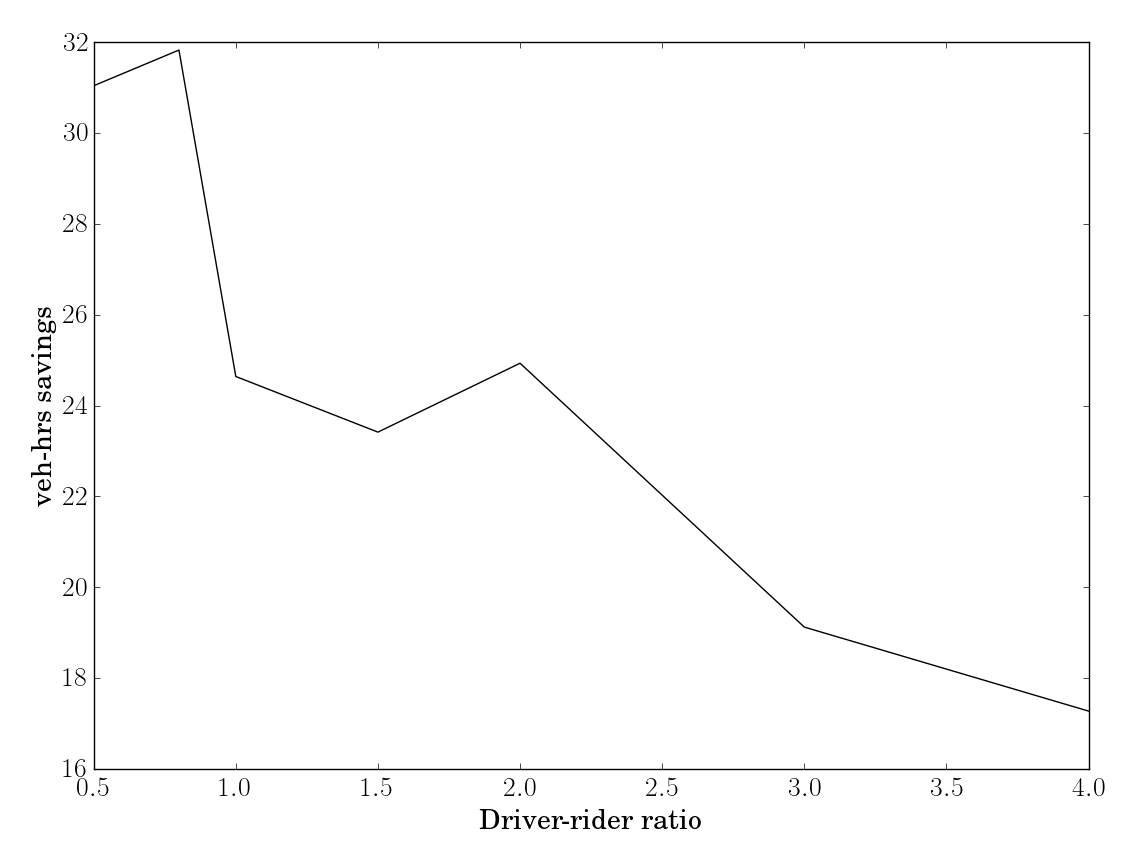}    
			\caption{Objective $Z_1$}
		\end{subfigure}        
		\hfill
		\begin{subfigure}[t]{0.50\textwidth}            
			\includegraphics[width=0.37\textheight]{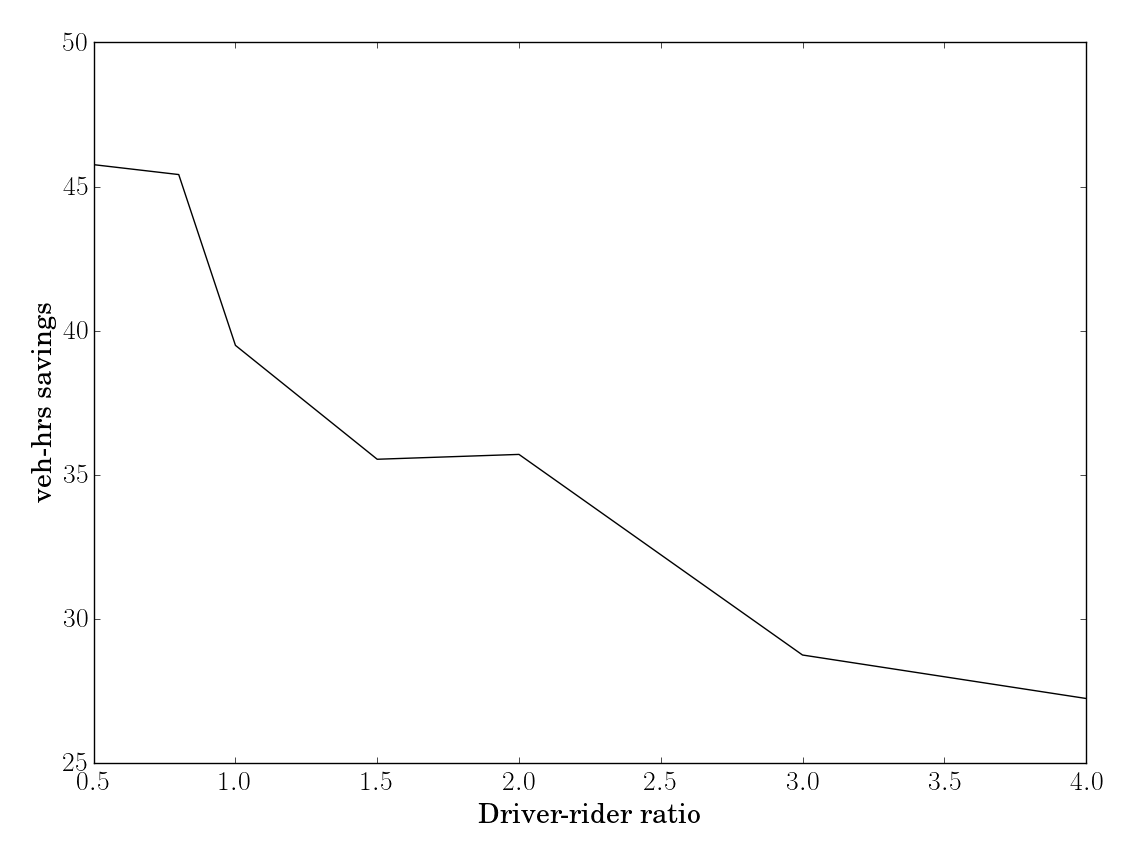}        
			\caption{Objective $Z_2$}
		\end{subfigure}
		\caption{The effect of driver-rider ratio on veh-hrs savings}
		\label{fig:vhs2}
	\end{figure}
	
	The effect of driver-rider ratio on the value of transit travel time components of matched participants is negligible after driver-rider ratio equal to 1.0. Overall, the average shared time between riders and drivers has a higher value in comparison to other components (Figure~\ref{fig:tt2}).
	\begin{figure}[h!]
		\begin{subfigure}[t]{0.50\textwidth}            
			\includegraphics[width=0.39\textheight]{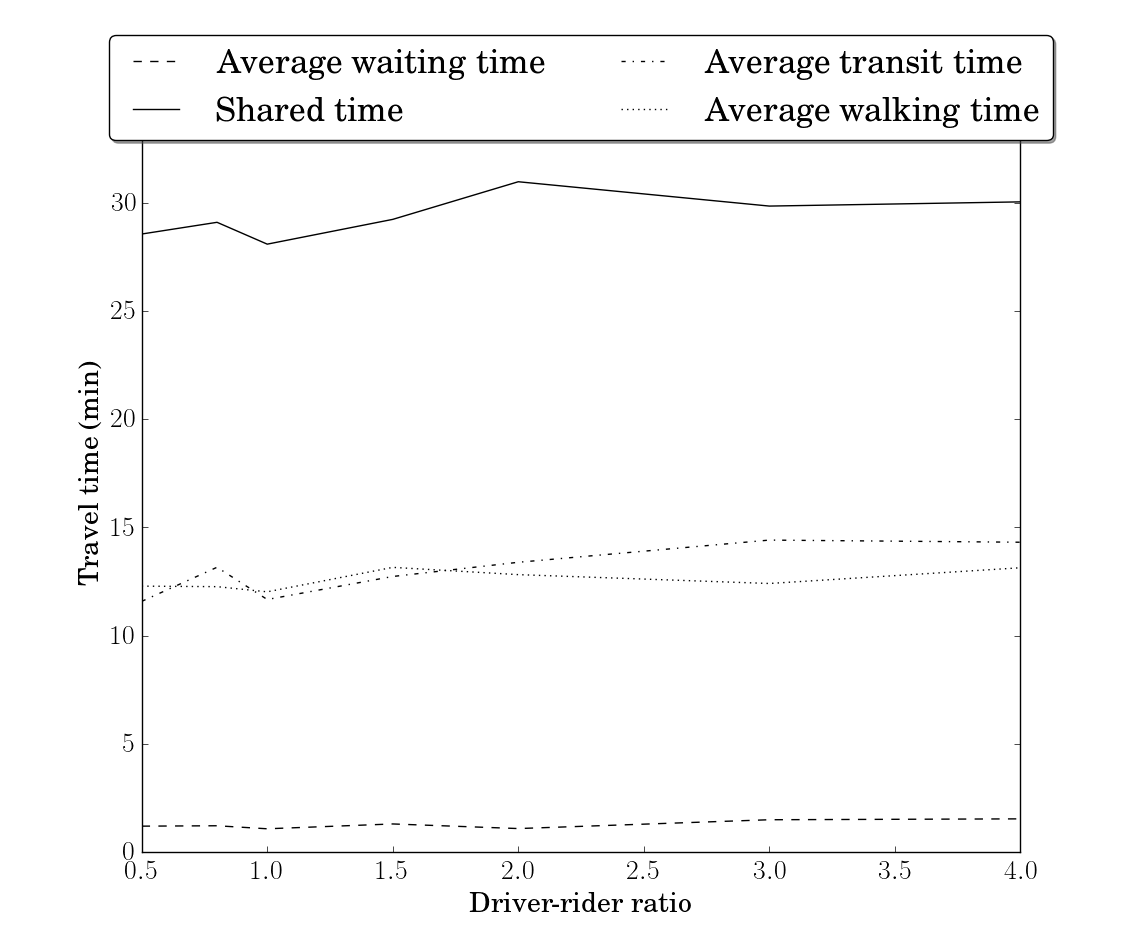}    
			\caption{Objective $Z_1$}
		\end{subfigure}        
		\hfill
		\begin{subfigure}[t]{0.50\textwidth}            
			\includegraphics[width=0.3952\textheight]{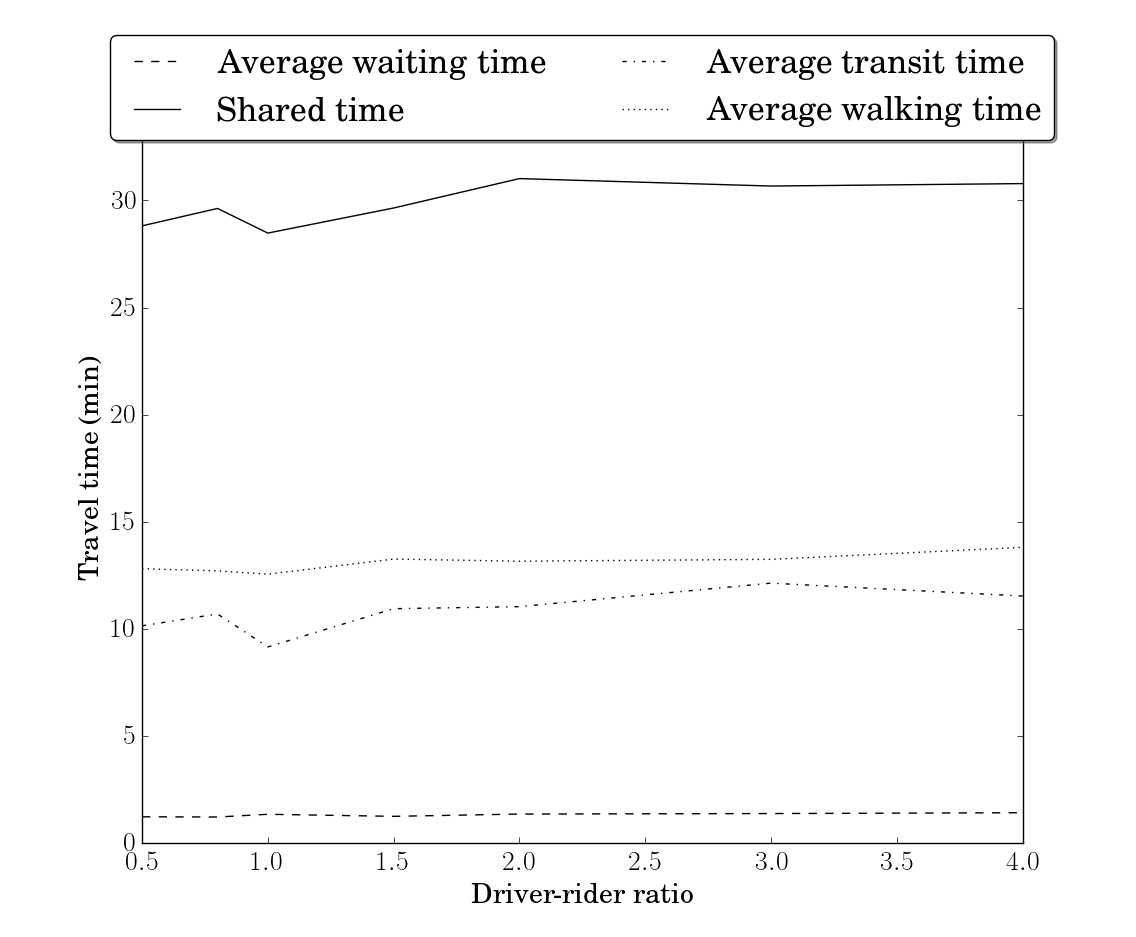}        
			\caption{Objective $Z_2$}
		\end{subfigure}
		\caption{The effect of driver-rider ratio on different time components of rider itinerary}
		\label{fig:tt2}
	\end{figure}

		\subsection{Combining transit-based ridesharing with stand-alone ridesharing}
		In \cref{sec:prelim}, we presented a hypothesis that combining the transit-based ridesharing with stand-alone ridesharing would improve the matching rate and overall savings in veh-hrs. The matching rate is defined as the number of matches divided by the average number of riders and drivers in the system. To get insights into this hypothesis, we run simulations in three different settings:
	\begin{enumerate}
	    \item     \textit{Stand-alone ridesharing (RS)}: In this setting, we find matches between riders and drivers, in which case, the complete trip of a rider is covered by the driver. We use the method proposed by \citet{Agatz2011} to compute the feasible matches and then solve the optimization program (\ref{eq:10}) to find optimal matches.
	
	\item     \textit{Stand-alone transit-based ridesharing (TRS)}: In this setting, we run Algorithm~\ref{alg:sbtsp} to find feasible matches and solve the optimization program (\ref{eq:10}) to find optimal matches.
	
	\item \textit{Combining stand-alone ridesharing and transit-based ridesharing (RS + TRS)}: In this setting, we allow drivers to cover the rider's complete trip. We find feasible matches using both Algorithm~\ref{alg:sbtsp}s and the method proposed by \citet{Agatz2011}, and solve the optimization program (\ref{eq:10}) to find optimal matches.
	\end{enumerate}
    For the driver-rider ratio equal to 1, the results are presented for all of the above settings in Table \ref{tab:6}. For 1035 drivers, and 1144 riders in the system, the optimal matches found in three different settings are 562, 192, and 575 respectively. Combining the RS with TRS increases the matching rate in the case of both objectives ($Z_1$ and $Z_2$). However, the savings in the case of $Z_2$ is significantly higher in comparison to $Z_1$. If we look at the results, we can see that the average savings per match by transit-based ridesharing setting are 21 veh-min in comparison to 14 veh-min in case of stand-alone ridesharing. The average savings is calculated by dividing the veh-hrs saving by the number of final matches. This shows that involving transit trips can significantly help in reducing congestion. Moreover, with better results, by combining both types of ridesharing settings, we have more confidence in the gain from adopting the proposed system and promoting transit.

	% \usepackage{multirow}
	\begin{table}[]
		\caption{Comparing matching results for three different setting of ridesharing}
		\label{tab:6}
		\resizebox{\textwidth}{!}{%
			\begin{tabular}{lllll}
				\cline{3-5}
				& \multicolumn{1}{l|}{}                                                     & \multicolumn{1}{l|}{RS}               & \multicolumn{1}{l|}{TRS}              & \multicolumn{1}{l|}{RS + TRS}              \\ \cline{2-5} 
				\multicolumn{1}{l|}{}                                    & \multicolumn{1}{l|}{No of participants}                                   & \multicolumn{1}{l|}{1,035 D, 1,144 R} & \multicolumn{1}{l|}{1,035 D, 1,144 R} & \multicolumn{1}{l|}{1,035 D, 1,144 R}      \\ \cline{2-5} 
				\multicolumn{1}{l|}{}                                    & \multicolumn{1}{l|}{Number of riders and drivers in the feasible matches} & \multicolumn{1}{l|}{925R 716 D}       & \multicolumn{1}{l|}{218R 344 D}       & \multicolumn{1}{l|}{936R, 720D}            \\ \cline{2-5} 
				\multicolumn{1}{l|}{}                                    & \multicolumn{1}{l|}{}                                                     & \multicolumn{1}{l|}{}                 & \multicolumn{1}{l|}{}                 & \multicolumn{1}{l|}{}                      \\ \hline
				\multicolumn{1}{|l|}{\multirow{7}{*}{Maximize savings}}  & \multicolumn{1}{l|}{Final matches found}                                  & \multicolumn{1}{l|}{562}              & \multicolumn{1}{l|}{192}              & \multicolumn{1}{l|}{575 (123 TRS, 452 RS)} \\ \cline{2-5} 
				\multicolumn{1}{|l|}{}                                   & \multicolumn{1}{l|}{Matching rate}                                        & \multicolumn{1}{l|}{51 \%}            & \multicolumn{1}{l|}{18 \%}            & \multicolumn{1}{l|}{53 \%}                 \\ \cline{2-5} 
				\multicolumn{1}{|l|}{}                                   & \multicolumn{1}{l|}{Vehicle-hrs savings}                                  & \multicolumn{1}{l|}{138}              & \multicolumn{1}{l|}{67}               & \multicolumn{1}{l|}{149}                   \\ \cline{2-5} 
				\multicolumn{1}{|l|}{}                                   & \multicolumn{1}{l|}{Average driver ride time (min)}                       & \multicolumn{1}{l|}{37.80}             & \multicolumn{1}{l|}{45.83}            & \multicolumn{1}{l|}{37.63}                 \\ \cline{2-5} 
				\multicolumn{1}{|l|}{}                                   & \multicolumn{1}{l|}{Average rider ride time (min)}                        & \multicolumn{1}{l|}{22.40}             & \multicolumn{1}{l|}{47.10}             & \multicolumn{1}{l|}{26.97}                 \\ \cline{2-5} 
				\multicolumn{1}{|l|}{}                                   & \multicolumn{1}{l|}{Driver detour time (min)}                             & \multicolumn{1}{l|}{7.66}             & \multicolumn{1}{l|}{8.82}             & \multicolumn{1}{l|}{7.47}                  \\ \cline{2-5} 
				\multicolumn{1}{|l|}{}                                   & \multicolumn{1}{l|}{Shared time(min)}                                     & \multicolumn{1}{l|}{22.40}             & \multicolumn{1}{l|}{28.10}             & \multicolumn{1}{l|}{22.55}                 \\ \cline{1-1} \cline{3-5} 
				\multicolumn{2}{l}{}                                                                                                                 &                                       &                                       &                                            \\ \hline
				\multicolumn{1}{|l|}{\multirow{7}{*}{Maximize matching}} & \multicolumn{1}{l|}{Final matches found}                                  & \multicolumn{1}{l|}{583}              & \multicolumn{1}{l|}{192}              & \multicolumn{1}{l|}{593 (104 TRS, 489 RS)} \\ \cline{2-5} 
				\multicolumn{1}{|l|}{}                                   & \multicolumn{1}{l|}{Matching rate}                                        & \multicolumn{1}{l|}{53 \%}            & \multicolumn{1}{l|}{18 \%}            & \multicolumn{1}{l|}{55 \%}                 \\ \cline{2-5} 
				\multicolumn{1}{|l|}{}                                   & \multicolumn{1}{l|}{Vehicle-hrs savings}                                  & \multicolumn{1}{l|}{109.73}           & \multicolumn{1}{l|}{55.34}            & \multicolumn{1}{l|}{114}                   \\ \cline{2-5} 
				\multicolumn{1}{|l|}{}                                   & \multicolumn{1}{l|}{Average driver ride time (min)}                       & \multicolumn{1}{l|}{39.52}            & \multicolumn{1}{l|}{48.86}            & \multicolumn{1}{l|}{39.62}                 \\ \cline{2-5} 
				\multicolumn{1}{|l|}{}                                   & \multicolumn{1}{l|}{Average rider ride time (min)}                        & \multicolumn{1}{l|}{20.43}            & \multicolumn{1}{l|}{46.15}            & \multicolumn{1}{l|}{23.90}                  \\ \cline{2-5} 
				\multicolumn{1}{|l|}{}                                   & \multicolumn{1}{l|}{Driver detour time (min)}                             & \multicolumn{1}{l|}{9.14}             & \multicolumn{1}{l|}{11.85}            & \multicolumn{1}{l|}{9.29}                  \\ \cline{2-5} 
				\multicolumn{1}{|l|}{}                                   & \multicolumn{1}{l|}{Shared time(min)}                                     & \multicolumn{1}{l|}{20.43}            & \multicolumn{1}{l|}{27.92}            & \multicolumn{1}{l|}{20.62}                 \\ \hline
			\end{tabular}
		}
	\end{table}

	\subsection{An experiment on dynamic ridesharing using rolling horizon policy}\label{sec:rhe}
    To deal with the dynamics of a transit-based ridesharing system, we presented a rolling horizon algorithm in \cref{sec:rh}. In this section, we prepare a simulation environment to see how the rolling horizon approach works. We assumed a time step, $\sigma = $ of 5 minutes for this experiment. Depending on the demand, $\sigma$ plays an important role in deciding if one iteration of Algorithm~\ref{alg:rh} can be solved within a time interval of $\sigma$. For this experiment, we consider driver-rider ratio equal to 1.0, rider, and driver flexibility to be equal to 150\% and 50\% respectively, and 1\% participation rate during the morning peak. Figure~\ref{fig:rh} shows simulation results from 5:00 AM to 9:00 AM. The computational time results are presented in  Figure~\ref{fig:rh}a. The maximum computational time of 4.23 min was observed at 7:20 AM when there were 412 participants in the system. This shows that Algorithm~\ref{alg:rh} can run successfully for a 1\% participation rate with given $\sigma$. The optimization time remains a fraction of second for each iteration and most of the time was spent in finding feasible matches using Algorithm~\ref{alg:sbtsp}. \\

	\begin{figure}[h!]
		\begin{subfigure}[t]{\textwidth}			
			\includegraphics[width=0.75\textheight]{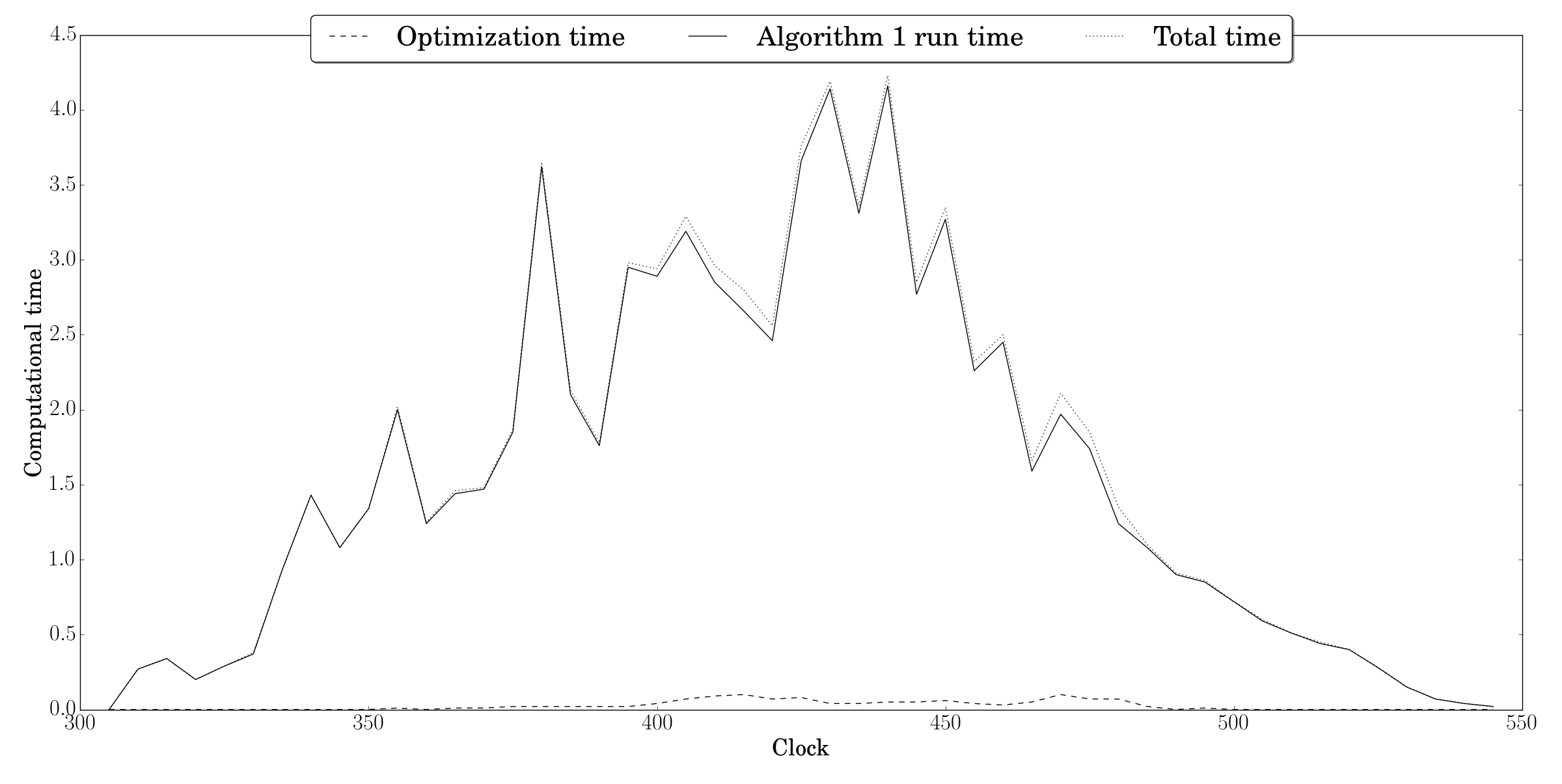}	
			\caption{Time spent during the simulation in each time interval}
		\end{subfigure}		
		
		\begin{subfigure}[t]{\textwidth}			
			\includegraphics[width=0.75\textheight]{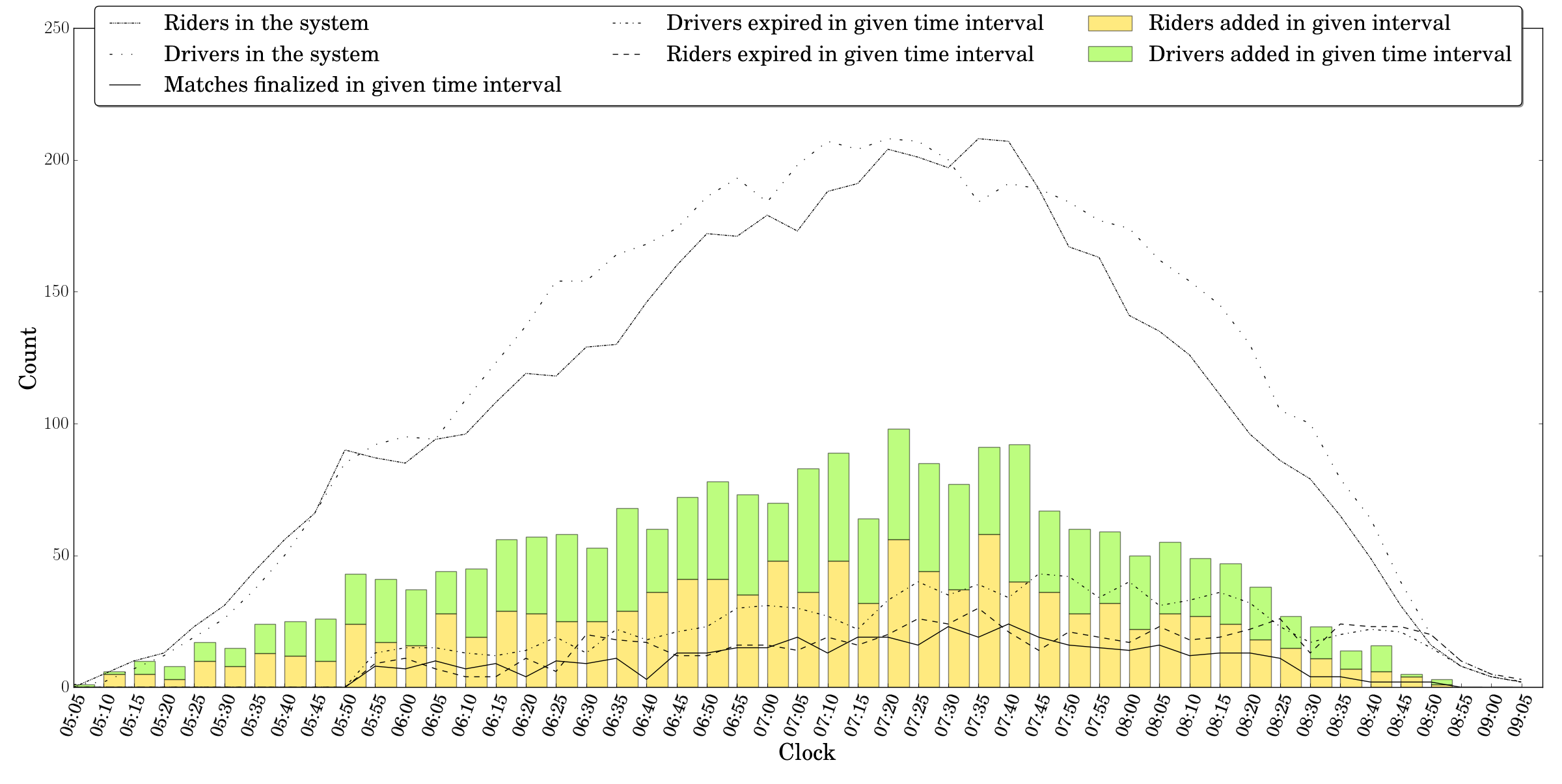}		
			\caption{Distribution of riders and drivers in the system}
		\end{subfigure}
		\caption{Algorithm~\ref{alg:rh} simulation results}
		\label{fig:rh}
	\end{figure}

    Figure~\ref{fig:rh}b shows the number of riders and drivers present and added in each interval. We can observe a peak demand between 7:00 AM-8:00 AM. The number of matches finalized in each time step is almost equal to the number of rider and driver requests being expired due to the unavailability of match satisfying given constraints. Overall, about 54\% of the participants found a successful match during the simulation. This example shows that the given algorithm can handle uncertainties arising in this dynamic matching problem.

	\section{Discussion}\label{sec:disc}
    The proposed research attempts to solve the transit FMLM problem using ridesharing under certain assumptions mentioned in \cref{sec:assp}. We discuss the possible implications of these assumptions and possibilities to resolve those issues. First, we assumed that transit vehicles have unlimited capacity. However, the congested transit systems often face the problem where passengers fail to board the train due to limited capacity. This will result in extra waiting time, which was not considered when finding feasible matches. A time-dependent origin-destination flow matrix (\cite{Nassir2011, Kumar2018b, Kumar2019}) can help in evaluating a failure to board probability to incorporate in the model. Second, it was assumed that the travel time on both road and transit network is reliable. However, the travel time experienced by passengers is often subject to uncertainty. This uncertainty is more prominent in the case of bus transit service. One way to address this problem is to create matches in the framework which allow drop-off only at stops served by high-frequency routes. The higher the frequency of transit service, lesser is the waiting time. Another way is to formulate this problem as an online shortest path problem, which accounts for the uncertainty in the travel time and provides optimal policies for each possible realization (\cite{Waller2002, Khani2019}). Future work is required to develop an algorithm that creates feasible matches based on online information.  Third, better calibration of service time based on the dwell time models is needed. 
	
	\section{Conclusions and future research}\label{sec:conc}
    This research proposes ridesharing as a promising solution to the transit FMLM problem. To integrate ridesharing with fixed-route transit, we developed a labeling algorithm to find feasible matches between drivers and riders. The algorithm uses a schedule-based transit shortest path to generate optimal itineraries for riders, which captures various complexities of a transit network such as precise waiting time, in-vehicle time, transfer time, and access time. To reduce the size of the network search, the algorithm uses the concept of space-time prism (STP) which provides constraints on rider and driver movement based on the available time budget and location. Using a matching optimization program, the riders and drivers are matched together up to the first mile of the rider trip. Two different objectives are considered for the problem; maximizing the total number of matches and maximizing the total vehicle-hrs savings. Using simulation experiments on real data from a large scale network, we found that the proposed ridesharing program can save a significant amount of veh-hrs spent in the system. The number of matches found using both objectives were found to be close to each other. However, maximizing veh-hrs savings can save a lot more veh-hrs than maximizing the total number of matches. This shows that maximizing the matching rate is not a good objective. We also observed different trip components of drivers and riders and found that driving from riders' origin to drop-off stations is a significant component of both riders and drivers' itinerary. This is due to less time flexibility in riders' schedule as transit travel time is generally higher than the corresponding driving time.  We also performed a sensitivity analysis on the participation rate, time flexibility, and driver-rider ratio. With an increase in time flexibility and participation rate, the runtime of the algorithm increases, producing a higher number of feasible matches. An equal number of riders and drivers in the system maximizes the matching rate. By combining the stand-alone ridesharing with the transit-based ridesharing program, we obtained a better matching rate and veh-hrs savings. We also proposed an algorithm to solve the given ridesharing matching problem in real-time using a rolling horizon approach. For a 1\% participation rate among travelers facing the FMLM problem, a time step of 5 minutes is adequate for solving the given problem on a regular desktop computer. Furthermore, computational time can be further reduced by using parallelization techniques. Using simulation experiments, we observed that about 46\% of riders and drivers did not find a suitable match due to the lack of a feasible partner or time flexibility. Twin Cities sub-urban region being large and sparse has a lower number of trips generated in those regions. The lower matching rate is due to less number of requests. With a higher number of requests (e.g., in denser urban areas such as Manhattan region), the matching rate can be higher.\\
	
    The current research can be expanded in multiple directions. To improve the matching rate of the current program, there is a need to develop a multiple-matching algorithm. One possibility is to sequentially grow the number of feasible riders for a given driver by checking different time constraints. It is also challenging to develop such an algorithm for riders facing both first and last mile problems which means there is no transit access at both ends of a transit trip. There is a need to develop an algorithm that can generate several ridesharing modal choices for users.

	\bibliographystyle{plainnat}
	\bibliography{library.bib}

\end{document}